\documentclass[twocolumn,english,aps,pra,showpacs,superscriptaddress]{revtex4-1}
\usepackage{amsmath}
\usepackage{amssymb}
\usepackage{stackrel}
\usepackage{graphicx}
\usepackage{esint}

\usepackage{epsfig}
\usepackage{graphicx}%
\usepackage{dcolumn}
\usepackage{amsmath}
\usepackage{longtable}
\usepackage[utf8]{inputenc}
\usepackage[english]{babel}
\usepackage[utf8]{inputenc}
\usepackage[english]{babel}

\usepackage[utf8]{inputenc}
\usepackage[english]{babel}
 \usepackage{braket}
\usepackage{color}
\usepackage{mathrsfs}

\usepackage{amsthm}
\usepackage{enumerate}
\usepackage[toc,page]{appendix}
\usepackage{setspace}
\usepackage{enumitem}
\usepackage{graphicx}
\newlist{steps}{enumerate}{1}
\setlist[steps, 1]{label = Step \arabic*:}

\makeatletter
\usepackage{babel}

\makeatother

 {
      \theoremstyle{plain}
      \newtheorem{assumption}{Assumption}
  }

\begin{document}

\title{Non-Markovianity, information backflow and system-environment correlation for open-quantum-system processes}

\author{Yun-Yi Hsieh}
\affiliation{Department of Physics and Center for Theoretical Physics, National
Taiwan University, Taipei 10617}

\author{Zheng-Yao Su}
\affiliation{National Center for High-Performance Computing, Hsinchu 300, Taiwan}

\author{Hsi-Sheng Goan}
\email{goan@phys.ntu.edu.tw}

\affiliation{Department of Physics and Center for Theoretical Physics, National
Taiwan University, Taipei 10617}

\affiliation{Center for Quantum Science and Engineering, National Taiwan University,
Taipei 10617, Taiwan}

\date{\today}
\begin{abstract}
  A Markovian process of a system is defined classically as a process in which the future state of the system is fully determined by only its present state, not by its previous history. There have been several  measures of 
  non-Markovianity to quantify the degrees of non-Markovian effect in a process of an open quantum system based on information backflow from the environment to the system.  However, the  condition for the witness of the system information backflow does not coincide with the classical definition of a Markovian process. Recently, a new measure with a condition that coincides with the classical definition in the relevant limit has been proposed. Here, we focus on the new definition (measure) for quantum non-Markovian processes, and characterize the Markovian condition as a quantum process that has no information backflow through the reduced environment state (IBTRES) and no system-environment correlation effect (SECE). The action of IBTRES produces non-Markovian effects by flowing the information of quantum operations performed by an experimenter at earlier times back to the system through the environment, while the SECE can produce non-Markovian effect without carrying any earlier quantum operation information. We give the necessary and sufficient conditions for no IBTRES and no SECE, respectively, and show that a process is Markovian if and only if it has no IBTRES and no SECE.
The quantitative measures and algorithms for calculating  non-Markovianity, IBTRES and soly-SECE are explicitly presented.

\end{abstract}

\maketitle 

\section{Introduction}

The Markovian approximation is a useful and often made approximation to describe the dynamics of both classical and quantum systems. This approximation makes the future system state depend on only the present one, not on its previous history, such that the system dynamics is easier to solve and described. 
In an open quantum system, the Markovian evolution equation of the reduced system density matrix used to be defined as the celebrated Lindblad master equation \cite{Lindblad}.
However, the Markovian approximation or master equation
breaks down in systems with strong coupling to their environments, with structured environment spectral densities, and at low temperatures. A non-Markovian dynamics can keep the previous system state information in the environment or/and in the quantum correlation between the system and environment, 
and then flows back the memory information to the system in the future.
This information backflow action breaks the Markovian approximation.
In recent years,
there have been many studies that tend to qualify
 the {\em memory effects} in quantum processes using different witness measures for non-Markovianity \cite{BLP,RHP,NMdegree,NMgeo,NM1,NM2,NM3,NM4,NM5,NM6,NM7,NM11,NM12,NMChen,NMChen,NMBbdini,NMZhang}, such as measures based on the {\em information backflow from the environment to the system} in \cite{BLP,RHP,NMdegree,NMgeo,NM2,NM3,NM4,NM5,NM6}, on the positivity of dynamical maps in \cite{RHP,NMdegree,NM7,NM11,NM12,NMChen}, on the conditional past-future independence \cite{NMBbdini}, on the two-time correlation functions \cite{NMZhang} and on the distance between dynamical maps and the map in experiential form, $\exp({\mathcal{L}})$, where $\mathcal{L}$ is a Lindblad operator \cite{NM1}. Several review articles on quantum non-Markovianity are available \cite{NM13,NM14,NM15,NMLocal}.

Previous measures \cite{BLP,RHP,NMdegree,NMgeo,NM1,NM2,NM3,NM4,NM5,NM6,NM7,NM11,NM12,NM13,NM14,NM15,NMChen,NMLocal,NMBbdini,NMZhang} are all witness of a non-Markovian process. They are sufficient but not necessary conditions for a non-Markovian process. 
A Markovian process in a classical system is described by the definition that the future state of the system depends on only the present state, not on the states in its history.
Recently, a new 
definition of a quantum Markovian process with operational interpretation 
was proposed \cite{KModi,KModi2}, and it coincides with the definition of a classical Markovian process.

In this paper, we make a 
connection of the new Markovian process definition with the process of information backflow.
We characterize a Markovian process as a quantum process that has no information backflow through the reduced environment state (IBTRES) and no system-environment correlation effect (SECE).
We find that IBTRES can produce the non-Markovian effect by sending the system state information at earlier time steps (history state information) to the final system state, while 
SECE can produce the non-Markovian effect without carrying any history state information. We give the operational
definition and the necessary and sufficient conditions for no IBTRES and no SECE. The necessary and sufficient conditions for no IBTRES and no SECE in the process tensor representation are also presented.
We show that a quantum process is Markovian if and only if it has no IBTRES and no SECE. 
Quantitative measures for non-Markovianity, IBTRES, and SECE,  and  algorithms to calculate these measures are presented and described.

The paper is organized as follows. 
In Sec.~II, we define and describe quantum maps and quantum operations in three representations (forms): the tensor representation, the map representation and the Choi matrix representation. This also enables us to introduce the notations we use in this paper. In Sec.~III, we show that a process is Markovian if and only if it has no-IBTRES and no-SECE. In Sec.IV, we review the concept and definition of process tensors.
In Sec.~V and Sec.~VI, we give the operational definitions (in terms of process tensors) and distance measures for non-Markovianity and IBTRES, and describe corresponding algorithms to calculate the distance measures.
In Sec.~VII, we present the operational definition (in terms of process tensors) and distance measures for SECE.
In Sec.~VIII, the definition and measure of solely SECE are introduced and
an algorithm for calculating the measure  quantitatively is presented.
In Sec.~IX, we compare the difference between SECE and IBTRES in a quantum process. Finally a short conclusion is given in Sec.~X.

\section{Representations and Notations}

In this section, we define the physical quantities and quantum maps in different representations or forms, and introduce also the notations we use in this paper.
The density matrices
of a single-party system  $\rho^{S} \in \mathcal{B}\left(\mathcal{H}^{S}\right)$ and a bipartite system-environment joint system $\rho^{SE}\in \mathcal{B}\left(\mathcal{H}^{SE}\right)$, where  $\mathcal{B}\left(\mathcal{H}^{S}\right)$  and  $\mathcal{B}\left(\mathcal{H}^{SE}\right)$  \cite{KModi3} are the corresponding Hilbert spaces of bounded operators of the single-party system and the bipartite system-environment joint system, respectively,  can be written as  
\begin{equation}
\rho^{S}=\sum_{ij}\rho^{S}_{ij}\ket{i}_{S}\bra{j}
\end{equation}
and
\begin{equation}
\rho^{SE}=\sum_{ij\alpha\beta}\rho^{SE}_{ij,\alpha\beta}\ket{i}_{S}\bra{j}\otimes\ket{\alpha}_{E}\bra{\beta}.
\end{equation}
Here we have used the Roman letters $\{|i\rangle\}$ to denote an orthonormal basis of the system Hilbert space and the Greek letters $\{|\alpha\rangle\}$
to denote an orthonormal basis of the environment Hilbert space. 
The tensor representations (forms) of these two matrices are $\rho^{S}_{ij}$ and $\rho^{SE}_{ij,\alpha\beta}$, respectively. 

A quantum map  $\mathcal{A}$ can be written 
$\mathcal{A}:\mathcal{B}\left(\mathcal{H}^{in}\right) \rightarrow \mathcal{B}\left(\mathcal{H}^{out}\right)$ as a mapping from bounded operators
on the input Hilbert space to bounded operators on the output Hilbert space.
Three representations or forms of
a quantum map  $\mathcal{A}$
are used in this paper, and they are the tensor form $A^{i_{0}j_{0}}_{i_{1}j_{1}}$, the Choi matrix form \cite{Choi} $A$ and the map form $\mathcal{A}$.
To be more precise, suppose a quantum map $\mathcal{A}$
is written as
\begin{equation}
\label{eq3}
\mathcal{A}\left[\ket{i_{0}}_{\rm in}\bra{j_{0}}\right]=\sum_{i_{1}j_{1}}A^{i_{0}j_{0}}_{i_{1}j_{1}}\ket{i_{1}}_{\rm out}\bra{j_{1}}, 
\end{equation}
then its tensor form is given by $A^{i_{0}j_{0}}_{i_{1}j_{1}}$, and its Choi matrix form is 
\begin{equation}
\label{AChoi}
A=\sum_{i_{0}j_{0}i_{1}j_{1}}A^{i_{0}j_{0}}_{i_{1}j_{1}}\ket{i_{1}}_{\rm out}\bra{j_{1}}\otimes\ket{i_{0}}_{\rm in}\bra{j_{0}},
\end{equation}
where the subscripts in and out indicate the input state and output state, respectively. These three forms or representations are isomorphic. The definition of a trace-preserving map is described in Appendix \ref{HPTP}.

A quantum operation is Hermiticity preserving, trace non-increasing and completely positive. Its corresponding Choi matrix must be a density matrix with a normalization factor $n_{d}$ and obey the trace preserving condition of Eq.~(\ref{trace preserving}). Here $n_{d}$ is the dimension of the Hilbert space that the map acts on. 

Let $\rho'^{S}$ be the result of the map $\mathcal{A}$ acting on $\rho^{S}$:
\begin{equation}
\rho'^{S}=\mathcal{A}[\rho^S]=\sum_{i_{0}j_{0}}A^{i_{0}j_{0}}_{i_{1}j_{1}}\rho^{S}_{i_{0}j_{0}}\ket{i_{1}}\bra{j_{1}},
\end{equation}
and its tensor representation is 
\begin{equation}
\rho'^{S}_{i_{1}j_{1}}=A^{i_{0}j_{0}}_{i_{1}j_{1}}\rho^{S}_{i_{0}j_{0}},
\end{equation}
where we have used the convention that repeated indices represent that the indices should be summed over.

A single-party map $\mathcal{A}$ acting on the system of a biparties system-environment density matrix $\rho^{SE}$ is written as $\mathcal{A}\otimes \mathcal{I}_{E}[\rho^{SE}]=\mathcal{A}[\rho^{SE}]$, where the identity map on the environment $\mathcal{I}_{E}$ is usually ignored. Following \cite{KModi,KModi2}, we restrict the quantum operations with notations $\mathcal{A}$,  $\Lambda_{\rho}$ and $\mathcal{T}_{[n]}$ to act only on the system space.

In this paper, $\mathcal{U}$ represents a total unitary map acting on both the system and environment Hilbert spaces, 
and we denote $\rho'^{SE}=\mathcal{U}[\rho^{SE}] = {\rm U}\rho^{SE}{\rm U}^{\dagger} $, where U is the unitary matrix operator corresponding to the map $\mathcal{U}$.
The tensor form of the process of the map is 
\begin{equation}
\rho'^{SE}_{i_{1}j_{1},\alpha_{1}\beta_{1}}=U^{i_{0}j_{0},\alpha_{0}\beta_{0}}_{i_{1}j_{1},\alpha_{1}\beta_{1}} \rho^{SE}_{i_{0}j_{0},\alpha_{0}\beta_{0}}.
\end{equation}
Note that $U^{i_{0}j_{0},\alpha_{0}\beta_{0}}_{i_{1}j_{1},\alpha_{1}\beta_{1}}$ is not the entries of unitary matrix operator U but the entries of the Choi matrix of the map $\mathcal{U}$:
\begin{equation}
\sum U^{i_{0}j_{0},\alpha_{0}\beta_{0}}_{i_{1}j_{1},\alpha_{1}\beta_{1}}\ket{i_{1}}\bra{j_{1}}\otimes\ket{i_{0}}\bra{j_{0}}\otimes\ket{\alpha_{1}}\bra{\beta_{1}}\otimes\ket{\alpha_{0}}\bra{\beta_{0}}.
\end{equation}

In our notation, the superscripts of the tensor form of a map correspond to the orthonormal basis indexes of the input Hilbert space, and the subscripts correspond to the orthonormal basis indexes of the output Hilbert space.
The columns of
The index pairs of superscripts and subscripts in the same column in the tensor form of a map correspond to the same system Hilbert space. For example,
 $(i_{0}j_{0})$ and $(i_{1}j_{1})$ [ $(\alpha_{0},\beta_{0}) $ and $(\alpha_{1},\beta_{1})$] in the same column in $U^{i_{0}j_{0},\alpha_{0}\beta_{0}}_{i_{1}j_{1},\alpha_{1}\beta_{1}}$ correspond to the input and output system ( environment) space. 
The indexes of a tensor are sorted from up to down and right to left in time. The corresponding Choi matrix basis indexes are sorted from right to left,
e.g.,
\begin{eqnarray*}
&&T^{i_{3}j_{3},i_{1}j_{1}}_{i_{4}j_{4},i_{2}j_{2},i_{0}j_{0}} \\
&& \rightarrow \sum T^{i_{3}j_{3},i_{1}j_{1}}_{i_{4}j_{4},i_{2}j_{2},i_{0}j_{0}} \\
&&\qquad  \times \ket{i_{4}}\bra{j_{4}} \otimes \ket{i_{3}}\bra{j_{3}} \otimes \ket{i_{2}}\bra{j_{2}} \otimes \ket{i_{1}}\bra{j_{1}} \otimes \ket{i_{0}}\bra{j_{0}} .
\end{eqnarray*}

We use the symbol "$\circ$" to represent the composition of maps. For example, $\mathcal{A'} \circ \mathcal{A}[\rho]$ means $\mathcal{A'} [\rho']$, where $\rho'=\mathcal{A}[\rho]$. Sometimes, when there is no confusion, we ignore the composite symbol "$\circ$" and write only $\mathcal{A'} \mathcal{A}[\rho]$.

\section{Quantum process}

Following Refs.~\cite{KModi,KModi2}, we make the following two assumptions,  
Assumptions \ref{AS1} and \ref{AS2},  about quantum operations that can be applied to the systems.

\begin{assumption}
For a composite quantum system composed of a system and an environment, it is assumed that an experimenter can apply quantum operations (e.g., $\mathcal{A}$) only on the system and also possibly on ancillary systems but not on the environment.
\label{AS1}
\end{assumption}
\begin{assumption}
Quantum operations (eg. $\mathcal{A}$ ) performed by an experimenter are assumed to be instantaneous in time, i.e., on a much shorter time scale than any other state dynamics.
\label{AS2}
\end{assumption}

\begin{figure}
\centering
\includegraphics[width=0.9\columnwidth]{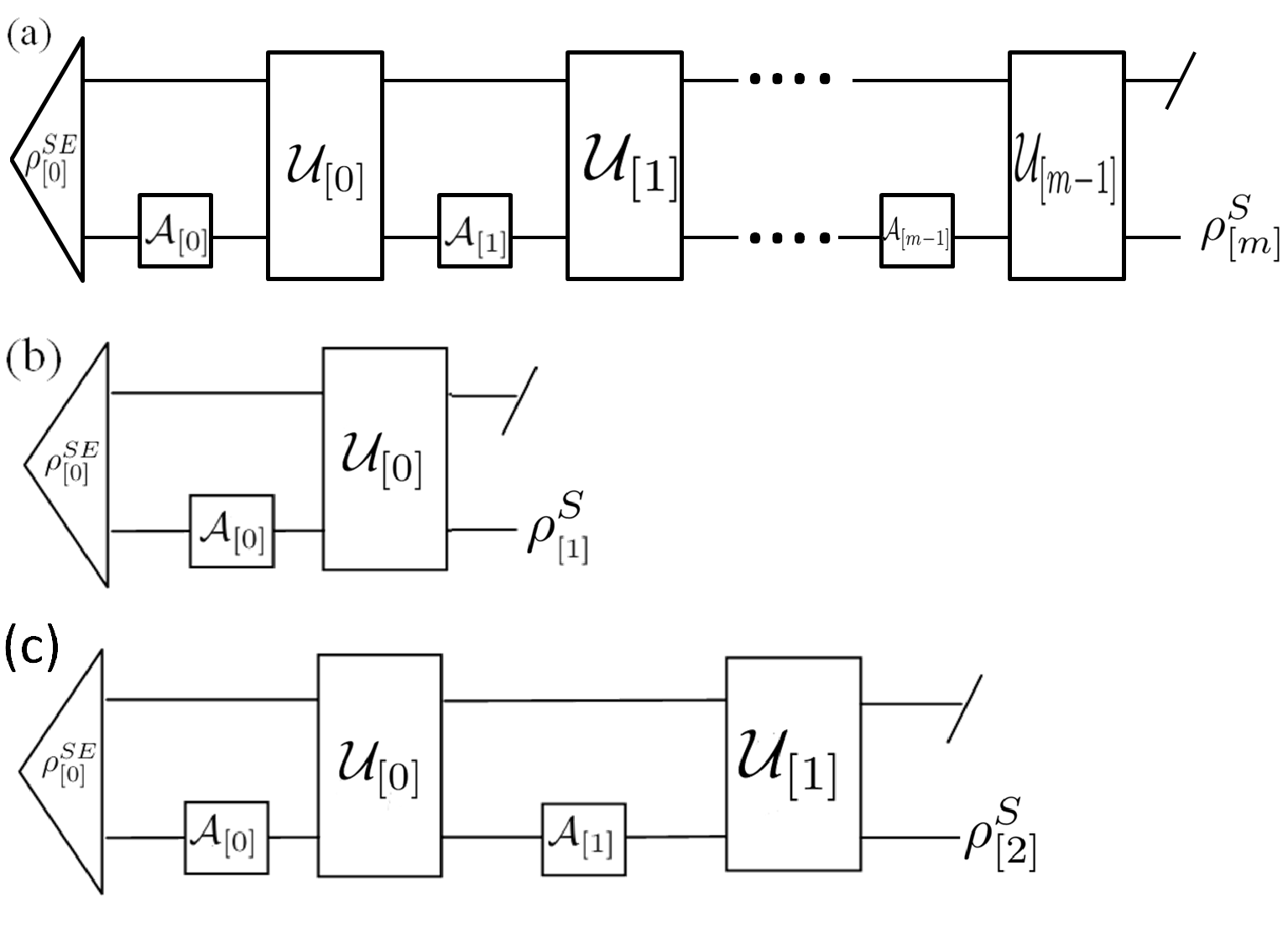}
\caption{Schematice illustrations of (a) an $m$-time-step process, (b) the first single-time-step process, and (c) the first two-time-step process.}
\label{STW}
\end{figure}

\newtheorem{mydef}{Definition}
\begin{mydef}
\label{mTSdef}
The reduced dynamic described by
\begin{equation}
\label{mTS}
\rho^{S}_{[m]}={\rm tr_E}\, \mathcal{U}_{[m-1]} \circ \mathcal{A}_{[m-1]}\circ \cdots\circ \mathcal{U}_{[1]} \circ \mathcal{A}_{[1]}\otimes\mathcal{U}_{[0]} \circ \mathcal{A}_{[0]}[\rho^{SE}_{[0]}]
\end{equation}
is an $m$-time-step process.
Here, $\rho^{SE}_{[0]} \in \mathcal{B}\left(\mathcal{H}^{SE}\right)$ is a fixed initial system-environment state. 
$\mathcal{U}_{[n]}: \mathcal{B}\left(\mathcal{H}^{SE}\right) \rightarrow \mathcal{B}\left(\mathcal{H}^{SE}\right)$ is the total unitary evolution map, and 
 $\mathcal{A}_{[n]} : \mathcal{B}\left(\mathcal{H}^{S}\right) \rightarrow \mathcal{B}\left(\mathcal{H}^{S}\right)$ is the quantum operation map at time step $n$
with $0 \leq n\leq m-1$.
\end{mydef}
A schematic illustration of an $m$-time-step process is shown in Fig.~\ref{STW}(a).
The first single-time-step process with $m=1$ in Eq.~(\ref{mTS}) is illustrated in Fig.~\ref{STW}(b) and described as
\begin{equation}
\label{STS}
\rho^{S}_{[1]} ={\rm tr_E}\, \mathcal{U}_{[0]} \circ \mathcal{A}_{[0]}[\rho^{SE}_{[0]}].
\end{equation}
In this paper, $\rho^{SE}_{[0]}$ is a fixed initial system-environment state.
Similarly, the first two-time-step process with $m=2$ in Eq.~(\ref{mTS}) is illustrated in Fig.~\ref{STW}(c) and described as
\begin{equation}
\label{2timeStep}
\rho^{S}_{[2]} ={\rm tr_E} \mathcal{U}_{[1]}\circ \mathcal{A}_{[1]} \circ \mathcal{U}_{[0]} \circ \mathcal{A}_{[0]}[\rho^{SE}_{[0]}].
\end{equation}
 The system-environment state at time step $n$ can be written as
\begin{equation}
\rho^{SE}_{[n]} =\mathcal{U}_{[n-1]} \circ \mathcal{A}_{[n-1]}[\rho^{SE}_{[n-1]}],  \ \  n=1,2,\cdots ,m ,
\end{equation}
and the reduced system state and reduced environment state of $\rho^{SE}_{[n]}$ are 
$\rho^{S}_{[n]}= {\rm tr_{E}}(\rho^{SE}_{[n]})$  and $\rho^{E}_{[n]}= {\rm tr_{S}}(\rho^{SE}_{[n]})$, respectively.
The $n$th time-step process, defined as the process from time step $(n-1)$ to time step $n$, in an $m$-time-step process is described by  
\begin{equation}
\rho^{S}_{[n]}={\rm tr_{E}} \, \mathcal{U}_{[n-1]} \circ \mathcal{A}_{[n-1]} [\rho^{SE}_{[n-1]} ],     \ \ \    n=1,2 \cdots ,m.
\label{nth_step}
\end{equation}

\begin{mydef}
\label{MarkovDef}
{\rm (Markovian process)}
A quantum system process is Markovian if and only if the future system state at next time step $\rho^{S}_{[n+1]}$ depends solely on the present one  $\mathcal{A}_{[n]}[\rho^{S}_{[n]}]$, no matter what quantum operations $\mathcal{A}_{[n-1]},\mathcal{A}_{[n-2]},\cdots, \mathcal{A}_{[0]}$ were previously applied to the system by an experimenter (i.e., independent of the state history of the process).
\end{mydef}

Definition \ref{MarkovDef} is equivalent to the definition of a Markovian process described in Ref.~\cite{KModi,KModi2}. 


\begin{mydef}
\label{no-SECEdef}
{\rm (No SECE)}
The system-environment correlation of $\rho^{SE}_{[n-1]}$ through the $n$th time-step process,
defined in Eq.~(\ref{nth_step}),
can not be detected by an experimenter 
 if and only if 
\begin{equation}
\label{no-SECEeq}
{\rm tr_E} \, \mathcal{U}_{[n-1]} \circ \mathcal{A}_{[n-1]} [\rho^{SE}_{[n-1]}]
 = \frac{{\rm tr_E} \, \mathcal{U}_{[n-1]} \circ \mathcal{A}_{[n-1]} [\rho^{S}_{[n-1]}\otimes\rho^{E}_{[n-1]}]}{{\rm tr} (\rho^{E}_{[n-1]})},
\end{equation}
where $\mathcal{A}_{[n-1]}$ is any quantum operation applied by the experimenter at time step $(n-1)$.
An $m$-time-step process has no SECE  if and only if it satisfies
Eq.~(\ref{no-SECEeq}) for all possible values $n$ of $1\leq n\leq m$.
\end{mydef}

This definition is obvious because if the the correlated system-environment state $\rho^{SE}_{[n-1]}$ on the left hand side of Eq.~(\ref{no-SECEeq})
replaced by the factorized state $\rho^{S}_{[n-1]}\otimes\rho^{E}_{[n-1]}$ as shown on the right hand side produces no difference, then there is no SECE. The denominator ${{\rm tr} (\rho^{E}_{[n-1]})}$ in  Eq.~(\ref{no-SECEeq}) is the normalization factor for the reduced environment state.

To discuss a process that has no IBTRES, let us define following trace-preserving maps:
\begin{equation}
\label{T0eq}
\mathcal{T}_{[0]}[\rho^{S}] \equiv \frac{{\rm tr_E} \mathcal{U}_{[0]} [\rho^{S} \otimes \rho^{E}_{[0]}]}{{\rm tr}(\rho^{E}_{[0]})}
\end{equation}
for the first time step, and for the $n$th time step
\begin{equation}
\label{Leq}
\mathcal{L}_{[n]} \left[\rho^{S},\mathcal{A}_{[n-1]} , \mathcal{A}_{[n-2]} ,\cdots, \mathcal{A}_{[0]} \right]  \equiv  \frac{{\rm tr_{E}}\mathcal{U}_{[n]}[\rho^{S} \otimes \rho^{E}_{[n]}]}{{\rm tr}(\rho^{E}_{[n]})} 
\end{equation}
with $n=1,2,\cdots,m-1$, where $\rho^{S} \in \mathcal{B}\left(\mathcal{H}^{S}\right)$ is an arbitary system state. 
Since the reduced environment state $\rho^{E}_{[n]}$ at time step $n$
in general depends on  quantum operations  ($\mathcal{A}_{[n-1]} , \mathcal{A}_{[n-2]} ,\cdots, \mathcal{A}_{[0]}$) at previous time steps, $\mathcal{L}_{[n]}$ should be a function of $(\mathcal{A}_{[n-1]} , \mathcal{A}_{[n-2]} ,\cdots, \mathcal{A}_{[0]})$, i.e., the information of the previous quantum operations
can flow back through the reduced environment state $\rho^{E}_{[n]}$.

\begin{mydef}
\label{no-IBTRESdef}
{\rm (No IBTRES)}
An $m$-time-step process with $m\geq 2$ has no IBTRES if and only if  
for all possible values  $n$ of $1\leq n\leq m-1$,
$\mathcal{L}_{[n]} \left[\rho^{S},\mathcal{A}_{[n-1]} , \mathcal{A}_{[n-2]} ,\cdots, \mathcal{A}_{[0]} \right]$ is
independent of all quantum operations {\rm (} $\mathcal{A}_{[n-1]} , \mathcal{A}_{[n-2]} ,\cdots, \mathcal{A}_{[0]}$ {\rm )}.
\end{mydef}
Note that no IBTRES can only be defined for an $m\geq 2$ process.
We say an $m$-times-step process has no IBTRES if and only if it obeys Definition \ref{no-IBTRESdef}.
If a map $\mathcal{L}_{[n]} \left[\rho^{S},\mathcal{A}_{[n-1]} , \mathcal{A}_{[n-2]} ,\cdots, \mathcal{A}_{[0]} \right]$ obeys Definition~\ref{no-IBTRESdef}, we set
\begin{equation}
\label{T1eq}
\mathcal{L}_{[n]} \left[\rho^{S},\mathcal{A}_{[n-1]} , \mathcal{A}_{[n-2]} ,\cdots, \mathcal{A}_{[0]} \right] = \mathcal{T}_{[n]}[\rho^{S}],
\end{equation}
and notice that $\mathcal{T}_{[0]}[\rho^{S}]$ is defined in Eq.~(\ref{T0eq}).
In Sec.\ref{secIBTRES}, we give an operational description of Definition \ref{no-IBTRESdef} for a two-time-step process that has no IBTRES.

\begin{mydef}
  \label{CMapDef}
{\rm (Constant map)}
$\Lambda_{\rho_{{\rm const}}}$ is a constant map if and only if
\begin{equation}
  \label{Cmap}
  \Lambda_{\rho_{\rm const}}[\rho]=\rho_{\rm const},
\end{equation}
 $\forall \rho \in \mathcal{B}\left(\mathcal{H}^{S_{0}}\right)$ with ${\rm tr}_{S_{0}}(\rho)=1$, and  $\rho_{{\rm const}} \in \mathcal{B}\left(\mathcal{H}^{S_{1}}\right)$ is a unique state (density matrix) with ${\rm tr}_{S_{1}}(\rho_{\rm const})=1$.
\end{mydef}
Note that in general $\mathcal{H}^{S_{0}}$ and $\mathcal{H}^{S_{1}}$ may not be the same.
The constant map maps an arbitrary density matrix $\rho$ with unit trace to a unique fixed density matrix $\rho_{\rm const}$ with unit trace. 
In Appendix \ref{constant_rep}, we give the tensor representation of the constant map.
A constant map has the action that
$\Lambda_{\rho_{{\rm const}}} \otimes \mathcal{I}_{E} [\rho^{SE}] =\rho_{{\rm const}} \otimes \rho^{E}$, $\forall \rho^{SE} \in \mathcal{B}\left(\mathcal{H}^{SE}\right)$, i.e., a constant map can destroy the system-environment correlation of $\rho^{SE}$ and erase the information of the system state (see Appendix \ref{ConstMapErase}).

\newtheorem{prop}{Proposition}
\begin{prop}
\label{prop1}
The first single-time-step process is Markovian if and only if it has no SECE.
\end{prop}
\begin{proof}

 Let us fist prove that if the first single-time-step process has no SECE, then it is Markovian.
 Taking 
 $n=1$ in Eq.~(\ref{no-SECEeq}) for the first single-time-step process of Eq.~(\ref{STS}), one obtains 
\begin{eqnarray}
\rho^{S}_{[1]} &=&
{\rm Tr_E} \, \mathcal{U}_{[0]} \circ \mathcal{A}_{[0]} [\rho^{SE}_{[0]}]\nonumber\\
 &=& \frac{{\rm Tr_E} \, \mathcal{U}_{[0]} \circ \mathcal{A}_{[0]} [\rho^{S}_{[0]}\otimes\rho^{E}_{[0]}]}{{\rm tr} (\rho^{E}_{[0]})},
 \label{no-SECEeq1}
\end{eqnarray}
where $\rho^{SE}_{[0]}$ is the fixed system-environment state at initial time $t=0$.
If  Eq.~(\ref{no-SECEeq1}) holds,  one can decompose the right hand side of Eq.~(\ref{no-SECEeq1}) into an operator sum form by using Kraus' theorem \cite{Kraus,Kraus2}:
\begin{eqnarray}
\frac{{\rm tr_E} \, \mathcal{U}_{[0]} \circ \mathcal{A}_{[0]} [\rho^{S}_{[0]}\otimes\rho^{E}_{[0]}]}{{\rm tr} (\rho^{E}_{[0]})} &=& \sum_{i} V_{i} \left( \mathcal{A}_{[0]}[\rho^{S}_{[0]}] \right)  V^{\dagger}_{i} \nonumber\\
&=& \mathcal{T}_{[0]} \circ \mathcal{A}_{[0]}[\rho^{S}_{[0]}] ,
\label{STmarkov}
\end{eqnarray}
where $V_{i}$ are the Kraus operators and depend only on $\mathcal{U}_{[0]}$ and $\rho^{E}_{[0]}$ but not on $\mathcal{A}_{[0]}$ and $\rho^{S}_{[0]}$. As a result, the output system state $\rho^{S}_{[1]}$ depends on only the system state $\mathcal{A}_{[0]}[\rho^{S}_{[0]}] $, i.e., the process is Markovian.

The reverse proof is as follows.
Let $\mathcal{A}_{[0]}$ and $\mathcal{A}'_{[0]}$ be two different quantum operations and their corresponding output states of
the first single-time-step process, 
Eq.~(\ref{STS}), be $\rho^{S}_{[1]}$ and $\rho'^{S}_{[1]}$, respectively.
Suppose that the process is Markovian, and then by
Definition~\ref{MarkovDef}, if 
\begin{equation}
\label{eqAeA}
\mathcal{A}_{[0]}[\rho^{S}_{[0]}]= \mathcal{A}'_{[0]}[\rho'^{S}_{[0]}],
\end{equation}
then
\begin{equation}
\label{SS'}
\rho^{S}_{[1]}= \rho'^{S}_{[1]}.
\end{equation}
One can set
\begin{equation}
\mathcal{A}'_{[0]}=\mathcal{A}_{[0]} \circ \Lambda_{\rho^{S}_{[0]}/{\rm tr} (\rho^{S}_{[0]})},
\end{equation}
where $\Lambda_{\rho^{S}_{[0]}/{\rm tr} (\rho^{S}_{[0]})}$ is a constant map which maps the input system state $\rho^{S}$ to a system state ${\rm tr}(\rho^{S})\rho^{S}_{[0]}/{\rm tr} (\rho^{S}_{[0]})$.
Then one obtains 
\begin{equation}
\label{idenityNEW}
 \mathcal{A}'_{[0]}[\rho^{S}_{[0]}]=\mathcal{A}_{[0]} \circ \Lambda_{\rho^{S}_{[0]}/{\rm tr} (\rho^{S}_{[0]})}[\rho^{S}_{[0]}] = \mathcal{A}_{[0]}[\rho^{S}_{[0]}],
\end{equation}
which is just Eq.~(\ref{eqAeA}).
Thus using Eq.~(\ref{SS'}) and Eq.~(\ref{STS}) yields 
\begin{eqnarray}
\hspace{-0.5cm}  {\rm tr_{E}} \, \mathcal{U}_{[0]} \circ \mathcal{A}_{[0]} [\rho^{SE}_{[0]} ]
&=& {\rm tr_{E}} \, \mathcal{U}_{[0]} \circ \mathcal{A}'_{[0]} [\rho^{SE}_{[0]} ]
\nonumber \\
\hspace{-0.5cm}&=&{\rm tr_E} \, \mathcal{U}_{[0]} \circ \mathcal{A}_{[0]}  \circ \Lambda_{\rho^{S}_{[0]}/{\rm tr} (\rho^{S}_{[0]})} [\rho^{SE}_{[0]}].                                  \label{rhorho'}
\end{eqnarray}
Applying the identity $\Lambda_{\rho^{S}_{[0]}/{\rm tr} (\rho^{S}_{[0]})} [\rho^{SE}_{[0]}] = \frac{\rho^{S}_{[0]}\otimes\ \rho^{E}_{[0]}}{{\rm tr} (\rho^{S}_{[0]})}$ to  Eq.~(\ref{rhorho'}), one obtains 
\begin{eqnarray}
\label{AnoCon}
 {\rm tr_{E}} \, \mathcal{U}_{[0]} \circ \mathcal{A}_{[0]} [\rho^{SE}_{[0]} ] 
&=& \frac{{\rm tr_E} \, \mathcal{U}_{[0]} \circ \mathcal{A}_{[0]} [\rho^{S}_{[0]} \otimes \rho^{E}_{[0]}]}{{\rm tr} (\rho^{S}_{[0]})} \nonumber \\
 &=& \frac{{\rm tr_E} \, \mathcal{U}_{[0]} \circ \mathcal{A}_{[0]} [\rho^{S}_{[0]} \otimes \rho^{E}_{[0]}]}{{\rm tr} (\rho^{E}_{[0]})},
\label{AtoCon}
\end{eqnarray}
where we have used the fact that ${\rm tr}(\rho^{E}_{[0]}) = {\rm tr}(\rho^{S}_{[0]})$ for the last equality.
Equation (\ref{AtoCon}) is just Eq.~(\ref{no-SECEeq1}), i.e., the definition of no SECE defined in Eq.~(\ref{no-SECEeq}) for $n=1$. 
\end{proof}

\begin{figure}
\includegraphics[width=0.9\columnwidth]{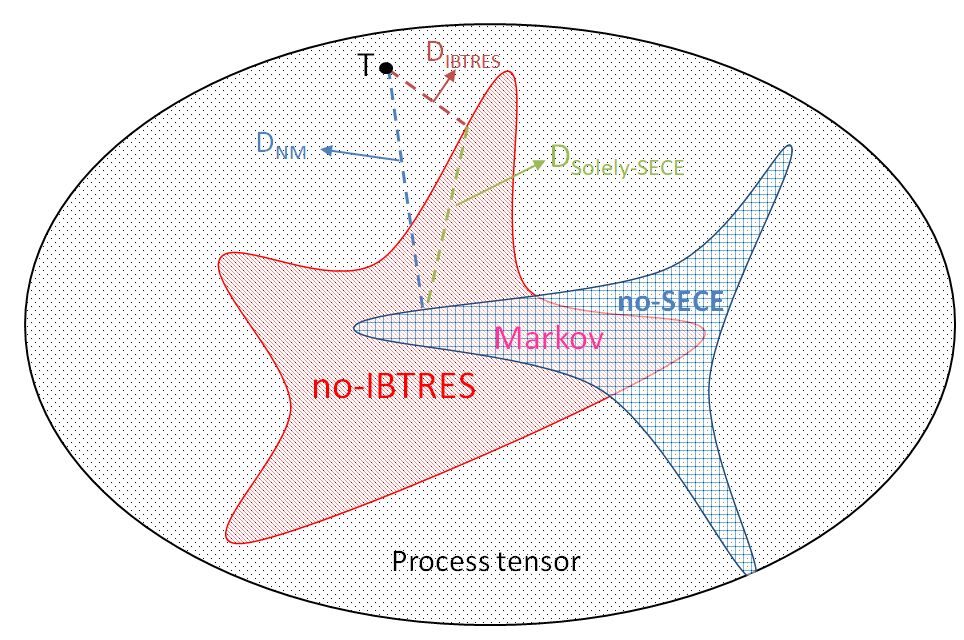}
\caption{Relationship among Markovian, no-SECE and no-IBTRES processes. The ellipse represents a collection of $m$-time-step processes defined in Definition \ref{mTSdef} with $m\geq 2$. The red set shaded by oblique lines represents a collection of processes with no IBTRES. The blue set shaded by vertical lines represents a collection of processes with no SECE. A Markovian process belongs to the intersection of the red set and the blue set. The non-Markovianity $D_{NM}$ and IBTRES $D_{\rm IBTRES}$ of a specific process $T$ can be defined as the minimum distances from $T$ to the Markovian set and to the no-IBTRES set, respectively. The distance $D_{NM}$, $D_{\rm IBTRES}$ and $D_{\rm solely-SECE}$ form a triangle.}
\label{set}
\end{figure}

\begin{prop}
\label{mainProp}
 An $m$-time-step process with $m \geq 2$ is Markovian if and only if the process has no SECE and no IBTRES. 
\end{prop}

\begin{proof}

The process from the time step $(n-1)$ to the time step $n$ is described by 
\begin{equation}
\label{QChanneln}
  \rho^{S}_{[n]} = {\rm tr_E} \mathcal{U}_{[n-1]} \mathcal{A}_{[n-1]}[\rho^{SE}_{[n-1]}],
\end{equation}
where $1\leq n\leq m$.	
Now, we consider the process that has no SECE. By Definition~\ref{no-SECEdef}, the $m$-time-step process must satisfy Eq.~(\ref{no-SECEeq}) for every time step $n$ of  $1\leq n\leq m$.
Thus the first time step process obtained by setting $n=1$ in Eq.~(\ref{QChanneln}) becomes
\begin{eqnarray}
\rho^{S}_{[1]}&=& {\rm tr_E} \, \mathcal{U}_{[0]} \circ \mathcal{A}_{[0]} [\rho^{SE}_{[0]}] \nonumber\\
&=&\frac{1}{{\rm tr} (\rho^{E}_{[0]})}\cdot{\rm tr_E} \, \mathcal{U}_{[0]} \circ \mathcal{A}_{[0]} [\rho^{S}_{[0]}\otimes\rho^{E}_{[0]}]\nonumber \\
&=&\mathcal{T}_{[0]} \circ \mathcal{A}_{[0]}[\rho^{S}_{[0]}]	
\end{eqnarray}
and similarly the $n$th time step ($2\leq n\leq m$) process becomes 
\begin{equation}
\begin{split}
  \rho^{S}_{[n]} =& {\rm tr_E} \mathcal{U}_{[n-1]}\circ \mathcal{A}_{[n-1]} [\rho^{SE}_{[n-1]}]  \\
	=& \frac{1}{\rho^{E}_{[n-1]}}{\rm tr_E} \mathcal{U}_{[n-1]}\circ \mathcal{A}_{[n-1]} [\rho^{S}_{[n-1]} \otimes \rho^{E}_{[n-1]}]  \\
	=&  \mathcal{L}_{[n-1]}[ \mathcal{A}_{[n-1]}[\rho^{S}_{[n-1]}], \mathcal{A}_{[n-2]}, \mathcal{A}_{[n-3]} , \cdots ,\mathcal{A}_{[0]}],
\end{split}
\label{n-step}
\end{equation}
where we have used the definitions of $\mathcal{T}_{[0]}[\rho^{S}]$ in Eq.~(\ref{T0eq}) and 
$\mathcal{L}_{[n]}$ in Eq.~(\ref{Leq}).

If the process has also no IBTRES, using Eq.~(\ref{T1eq}) that no previous information of the quantum operations flows back from $\rho^{E}_{[n-1]}$ to the system, one obtains  from Eq.~(\ref{n-step})  the $n$th time step ($2\leq n\leq m$) process  
\begin{equation}
  \rho^{S}_{[n]} = \mathcal{T}_{[n-1]} \circ \mathcal{A}_{[n-1]}[\rho^{S}_{[n-1]}].
 \label{Tn-1}
\end{equation}
By writing out the expression of Eq.~(\ref{Tn-1}) explicitly step by step for each $n$ of $1\leq n\leq m$, the total evolution process, Eq.~(\ref{mTS}), becomes
\begin{equation}
\label{MyM}
\rho^{S}_{[m]}=\mathcal{T}_{[m-1]} \circ \mathcal{A}_{[m-1]}\cdots\circ\mathcal{T}_{[1]} \circ \mathcal{A}_{[1]} \circ \mathcal{T}_{[0]} \circ \mathcal{A}_{[0]}[\rho^{S}_{[0]}],
\end{equation}
i.e., in a divisible form. 
Since ${\mathcal{T}_{[n]}}$ is defined in Eq.~(\ref{T1eq}) as a map independent of previous quantum operations, Eq.~(\ref{MyM}) describes a Markovian process because the future system state $\rho^{S}_{[n+1]}$ depends on only the previous system state $\mathcal{A}_{[n]}[\rho^{S}_{[n]}]$.

The proof for the reverse statement is shown below. To see whether a process has IBTRES, we set the quantum operation $\mathcal{A}_{[n]} =\Lambda_{\rho^{S}/{\rm tr}(\rho^{S})}$, where $\Lambda_{\rho^{S}/{\rm tr}(\rho^{S})}$ is a constant map and $\rho^{S}$ is an arbitrary system state. Then we obtain
\begin{eqnarray}
\rho^{S}_{[n]}&=&{\rm tr_{E}}\mathcal{U}_{[n-1]}\Lambda_{\rho^{S}/{\rm tr}(\rho^{S})}[\rho^{SE}_{[n-1]}] \nonumber\\
&=&{\rm tr_{E}}\mathcal{U}_{[n-1]}\circ[\rho^{S}\otimes\rho^{E}_{[n-1]}]/{\rm tr}(\rho^{S}) \nonumber\\
&=&{\rm tr_{E}}\mathcal{U}_{[n-1]}\circ[\rho^{S}\otimes\rho^{E}_{[n-1]}]/{\rm tr}(\rho^{E}) \nonumber\\
&=&\mathcal{L}_{[n-1]}\left[\rho^{S}, \mathcal{A}_{[n-2]},\mathcal{A}_{[n-3]},\cdots,\mathcal{A}_{[0]} \right]
\end{eqnarray}
If the process is Markovian, by Definition \ref{MarkovDef}, $\mathcal{L}_{[n-1]}$ does not depend on ($\mathcal{A}_{[n-2]}, \mathcal{A}_{[n-3]} ,\cdots, \mathcal{A}_{[0]}$). Therefore the process has no IBTRES by Definition \ref{no-IBTRESdef}. 
To show that a Markovian process has no SECE, the skill and procedure we used in the proof of  Proposition \ref{prop1} to show the first single-time-step Markovian process has no SECE
can be applied directly to a general single-time-step process of Eq.~(\ref{QChanneln}). In other words, by the replacement of $1\rightarrow n$ and $0\rightarrow n-1$ for the expression from
Eq.~(\ref{AnoCon}) to Eq.~(\ref{AtoCon}), one can show that a Markovian process has no SECE for a general single time step process.
Consequently, a Markovian process  must have no SECE for every time step process. 
\end{proof}

Proposition \ref{mainProp} is the most important result in this paper.

We have characterized an $m$-time-step ($m\geq 2$) Markovian process by two process sets: \{  no-SECE \} and \{ no-IBTRES \}. These two sets and their relations to the set of Markovian processes are schematically illustrated in Fig.~\ref{set} and summarized below.
\begin{itemize}
\item Intersection of  the \{  no-SECE \} and  \{  no-IBTRES \} sets is the Markovian set.
\item  The non-Markovian effect in the difference of the two sets
  \{ no-SECE \} and \{ Markovian \}, \{ no-SECE \} $-$ \{ Markovian \},
  is due to solely IBTRES. 
\item The non-Markovian effect in the difference  of the two sets
  \{ no-IBTRES \} and \{ Markovian \}, \{ no-IBTRES \} $-$ \{ Markovian \},  is due to solely SECE.
\end{itemize}

In the following sections, we will discuss the properties of these process sets in terms of process tensors. We give the necessary and sufficient conditions for no IBTRES and no SECE, respectively, and present explicitly 
the quantitative measures and algorithms for calculating non-Markovianity, IBTRES and soly-SECE.

\section{ Process tensor representation for single-time-step process and two-time-step process}
\label{SecPT}

In this section, we first argue that we can check whether a process is Markovian by investigating a two-time-step process with the time in the middle tunable from the initial time to the final time.
After that, review briefly the properties of process tensors of a single-time-step process and a two-time-step process in Refs.~\cite{KModi,KModi2} as we will use the process tensors to represent iff conditions for no IBTRES and no SECE, and to describe the qualitative measures and optimization algorithms.

For a quantum process with a fixed initial time $t_{0}$ and a final time $t_{2}$, it seems that one has to show that the process can be written in a divisible form as in Eq.~(\ref{MyM}) with $m \rightarrow \infty$. This is because showing an $m$-time-step process being Markovian for a finite $m$ does not guarantee an $(m+1)$-time-step process being Markovian. However, we can verify whether the process is Markovian with only a two-time-step process by checking whether the process satisfies the following Markovian condition: 
\begin{equation}
\label{only2}
\begin{split}
\rho^{S}_{[2]}(t_{2}) =&  {\rm tr_E}  \mathcal{U}_{[1]}(t_{2},t_{1}) \circ \mathcal{A}_{[1]} \circ \mathcal{U}_{[0]}(t_{1},t_{0}) \circ \mathcal{A}_{[0]} [\rho^{SE}_{[0]}(t_{0})] \\
&=\mathcal{T}_{[1]}(t_{2},t_{1}) \circ \mathcal{A}_{[1]} \circ \mathcal{T}_{[0]}(t_{1},t_{0})\circ \mathcal{A}_{[0]}[\rho^{S}_{[0]}(t_{0})].
\end{split}
\end{equation}
with the time $t_{1}$ in the middle continually varied from $t_{0}$ to $t_{2}$. Therefore, in the following part of this paper, we consider only the two-time-step process from $t_{0}$ to $t_{1}$ and $t_{1}$ to $t_{2}$.
We will sometimes abbreviate $\mathcal{U}_{[n]}(t_{n},t_{n-1})$ as $\mathcal{U}_{[n]}$ and $\mathcal{T}_{[n]}(t_{n},t_{n-1})$  as $\mathcal{T}_{[n]}$ for $n=1,2$.

\subsection{Process tensor for the first single-time-step process}
The process tensor is a representation for an open quantum system process and is a mapping from the sequence of quantum operations to the final system state \cite{KModi,KModi2}.
The final output system state $\rho^{S}_{[1]}$ in  Eq.~(\ref{STS}) can be written as 
 a function of the quantum operation $\mathcal{A}_{[0]}$ :
\begin{equation}
\label{Mfull}
  \rho^{S}_{[1]} = {\rm tr_E}  \mathcal{U}_{[0]} \mathcal{A}_{[0]} [\rho^{SE}_{[0]}] \equiv \mathcal{M}[\mathcal{A}_{[0]}],
\end{equation}
where $\mathcal{M}$ is called the map form of the process tensor for the first single-time-step process.

The tensor form of Eq.~(\ref{Mfull}) can be written as
\begin{equation}
\rho^{S}_{[1];{\color{red}{i_{1}j_{1}}}} = \delta_{\color{blue}{\alpha_{1}\beta_{1}}} 
U_{[0];{\color{red}_{i_{1}j_{1}}}_{,}{\color{blue}_{\alpha_{1}\beta_{1}}}}^{\  \   {\color{red}_{i_{0'}j_{0'}}},{\color{blue}_{\alpha_{0}\beta_{0}}}} A_{[0]; {\color{red} i_{0'}j_{0'}}}^{\  \  \ {\color{red}_{i_{0}j_{0}}}}\rho^{SE}_{[0]{\color{red}i_{0}j_{0}}{\color{blue}\alpha_{0}\beta_{0}}}
\label{Mtensor}
\end{equation}
Here, $\rho^{SE} \in \mathcal{B} \left( \mathcal{H}^{S_{0}E} \right)$, $A \in \mathcal{B} \left( \mathcal{H}^{S_{0'}S_{0}} \right)$ and $M \in \mathcal{B} \left( \mathcal{H}^{S_{1}S_{0'}S_{0}} \right)$ with $\mathcal{H}^{S_{0}}=\mathcal{H}^{S}$ being the Hilbert space of the system state.
We sort the indexes by the following way $(\cdots,n',n,\cdots,1',1,0',0)$ such that the quantum operation $A_{[n];i_{n'}j_{n'}}^{\ \ \ i_{n}j_{n}}$ with index $n$ carries the indexes $i_{n}$, $j_{n}$, $i_{n'}$ and $j_{n'}$ in a multi-time-step  process.
The Kronecker delta symbol $\delta_{\color{blue}{\alpha_{1}\beta_{1}}}$  corresponds to a trace operation over the environment degrees of freedom.
The tensor representation of the process tensor is to just remove $A_{[0]; {\color{red} i_{0'}j_{0'}}}^{\  \  \ {\color{red}_{i_{0}j_{0}}}}$ in Eq.~(\ref{Mtensor}) and is 

\begin{equation}
M^{{\color{red}{i_{0'}j_{0'}}}}_{{\color{red}{i_{1}j_{1},i_{0}j_{0}}}} 
= \delta_{\color{blue}{\alpha_{1}\beta_{1}}} 
U_{[0];{\color{red}_{i_{1}j_{1}}}_{,}{\color{blue}_{\alpha_{1}\beta_{1}}}}^{\  \   {\color{red}_{i_{0'}j_{0'}}},{\color{blue}_{\alpha_{0}\beta_{0}}}} \rho^{SE}_{[0]{\color{red}i_{0}j_{0}}{\color{blue}\alpha_{0}\beta_{0}}}.
\end{equation}

\subsection{Process tensor for the two-time-step process}
The final output system state $\rho^{S}_{[2]}$ of the two-time-step process Eq.~(\ref{2timeStep}) can be written as 
 a function of the quantum operations $\mathcal{A}_{[0]}$  and $\mathcal{A}_{[1]}$:
\begin{equation}
\label{eq34}
  \rho^{S}_{[2]} = {\rm tr_E} \mathcal{U}_{[1]} \mathcal{A}_{[1]} \mathcal{U}_{[0]} \mathcal{A}_{[0]} [\rho^{SE}_{[0]}] \equiv \mathcal{T}[\mathcal{A}_{[1]},\mathcal{A}_{[0]}],
\end{equation}
where $\mathcal{T}$ is the map form of a two-time-step process tensor that 
sends $\mathcal{A}_{[0]}$, $\mathcal{A}_{[1]}$ to the final state $\rho^{S}_{[2]}$.
In tensor form, Eq.~(\ref{eq34}) becomes
\begin{equation}
\begin{split}
\rho^{S}_{[2];{\color{red}{i_{2}j_{2}}}} &= \delta_{\color{blue}{\alpha_{2}\beta_{2}}} U_{[1];{\color{red}_{i_{2}j_{2}}}_{,}{\color{blue}_{\alpha_{2}\beta_{2}}}}^{\  \   {\color{red}_{i_{1'}j_{1'}}},{\color{blue}_{\alpha_{1}\beta_{1}}}} A_{[1]; {\color{red} i_{1'}j_{1'}}}^{\  \  \ {\color{red}_{i_{1}j_{1}}}} \\
&\cdot U_{[0];{\color{red}_{i_{1}j_{1}}}_{,}{\color{blue}_{\alpha_{1}\beta_{1}}}}^{\   \  {\color{red}_{i_{0'}j_{0'}}},{\color{blue}_{\alpha_{0}\beta_{0}}}} A_{[0]; {\color{red} i_{0'}j_{0'}}}^{\  \  \ {\color{red}_{i_{0}j_{0}}}}\rho^{SE}_{[0]{\color{red}i_{0}j_{0}}{\color{blue}\alpha_{0}\beta_{0}}}, 
\end{split}
\label{SEevo}
\end{equation}
The tensor representation of the process tensor is just to remove $A_{[0]; {\color{red} i_{0'}j_{0'}}}^{\  \  \ {\color{red}_{i_{0}j_{0}}}}$ and $A_{[1]; {\color{red} i_{1'}j_{1'}}}^{\  \  \ {\color{red}_{i_{1}j_{1}}}}$ in Eq.~(\ref{SEevo}):

\begin{equation}
\label{eq36}
T^{{\color{red}{i_{1'}j_{1'},i_{0'}j_{0'}}}}_{{\color{red}{i_{2}j_{2} \ , i_{1}j_{1} \ ,i_{0}j_{0}}}} = \delta_{\color{blue}{\alpha_{2}\beta_{2}}} U_{[1];{\color{red}_{i_{2}j_{2}}}_{,}{\color{blue}_{\alpha_{2}\beta_{2}}}}^{ \  \  {\color{red}_{i_{1'}j_{1'}}},{\color{blue}_{\alpha_{1}\beta_{1}}}} 
\cdot U_{[0];{\color{red}_{i_{1}j_{1}}}_{,}{\color{blue}_{\alpha_{1}\beta_{1}}}}^{ \  \  {\color{red}_{i_{0'}j_{0'}}},{\color{blue}_{\alpha_{0}\beta_{0}}}} \rho^{SE}_{[0]{\color{red}i_{0}j_{0}}{\color{blue}\alpha_{0}\beta_{0}}}
\end{equation}

In matrix representation, $T \in \mathcal{B}\left(\mathcal{H}^{S_{2}S_{1'}S_{1}S_{0'}S_{0}}\right)$ and is written as
\begin{equation}
\label{Eq39}
T=\sum T^{i_{1'}j_{1'},i_{0'}j_{0'}}_{i_{2}j_{2},i_{1}j_{1},i_{0}j_{0}}\ket{i_{2}i_{1'}i_{1}i_{0'}i_{0}}\bra{j_{2}j_{1'}j_{1}j_{0'}j_{0}}.
\end{equation}
with $\mathcal{H}^{S_{0}}=\mathcal{H}^{S}$ being the Hilbert space of the system state.

If initially, the system and environment are uncorrelated or in a factorized state, i.e., $\rho^{SE}_{[0]}=\rho^{S}_{[0]}\otimes \rho^{E}_{[0]}$, we can define a reduced process tensor $\tilde{T}$ in matrix form for two-time-step process as
\begin{equation}
\label{eq41}
T = \tilde{T} \otimes \rho^{S}_{[0]},
\end{equation}
where $\tilde{T} \in \mathcal{B}\left(\mathcal{H}^{S_{2}S_{1'}S_{1}S_{0'}}\right)$.
Using Eq.~(\ref{eq36}), we have its tensor form
\begin{equation}
\tilde{T}^{{\color{red}{i_{1'}j_{1'},i_{0'}j_{0'}}}}_{{\color{red}{i_{2}j_{2} \ , i_{1}j_{1} }}} = \delta_{\color{blue}{\alpha_{2}\beta_{2}}} U_{[1];{\color{red}_{i_{2}j_{2}}}_{,}{\color{blue}_{\alpha_{2}\beta_{2}}}}^{ \  \  {\color{red}_{i_{1'}j_{1'}}},{\color{blue}_{\alpha_{1}\beta_{1}}}} 
\cdot U_{[0];{\color{red}_{i_{1}j_{1}}}_{,}{\color{blue}_{\alpha_{1}\beta_{1}}}}^{ \  \  {\color{red}_{i_{0'}j_{0'}}},{\color{blue}_{\alpha_{0}\beta_{0}}}} \rho^{E}_{[0]{\color{blue}\alpha_{0}\beta_{0}}}.
\end{equation}
The reduced process tensor is useful to simplify algorithms (discussed later) for calculating various distance measures for an initial factorized system-environment state.

\section{Markovian process for a two-time-step quantum process}
\label{MP}
In this section, we give the operational definition for Markovian process, the distance measure for non-Markovianity and an optimization algorithm for calculating the distance measure.

\subsection{Operational definition}
A two-time-step process is Markovian if and only if the process tensor has following form \cite{KModi,KModi2}
\begin{equation}
\label{MarkovForm}
T = T_{[1]}\otimes T_{[0]} \otimes \rho^{S}_{[0]},
\end{equation}
where $T_{[0]}$ and $T_{[1]}$ are the Choi matrices of $\mathcal{T}_{[0]}$ and $\mathcal{T}_{[1]}$ in Eq.~(\ref{MyM}), respectively.
The operation description and the condition that the process tensor of a Markovian process must be in a product form of Eq.~(\ref{MarkovForm}) had been given in \cite{KModi,KModi2}.
In this paper, we go one step further to present how to construct the maps $\mathcal{T}_{[0]}$ and $\mathcal{T}_{[1]}$ [see Eq.~(\ref{T0eq}), Eq.~(\ref{T1eq}) and Eq.~(\ref{Leq})
and characterize the Markovian process as having no IBTRES and no SECE. The no IBTRES is strongly connected to the information backflow measures for non-Markovianity discussed in the literature \cite{BLP,RHP,NMdegree,NMgeo,NM2,NM3,NM4,NM5,NM6}.

\subsection{ Distance  measure for  non-Markovianity}
The non-Markovianity for the two-time-step process is defined as the minimum distance between the process tensor $T(t_{2},t_{1},t_{0})$ and all possible Markovian process tensors $T_{\rm Markov}$ \cite{KModi,KModi2} as
\begin{equation}
D_{NM}(t_{2},t_{1},t_{0})=\min_{{\rm T_{Markov}\in product\; form}} \mathcal{D}\left(T(t_{2},t_{1},t_{0}) , T_{\rm Markov} \right),
\label{eq:non-intNM}
\end{equation}
where $\mathcal{D}$ is a distance measure.
The non-Markovianity $\mathcal{N}(t_{2})$ defined in \cite{BLP,RHP} is given by 
\begin{equation}
\mathcal{N}(t_{2})= \int_{t_{0}=0}^{t_{2}} n(t_{1})dt_{1},
\end{equation}
where $n(t_{1})$ is the non-Markovianity in the time interval $[t_{1},t_{1}+dt_{1}]$, and $\mathcal{N}(t_{2})$ is the total integrated effect for $t_{1}$    from $t_{0}=0$ to $t_{2}$. The non-Markovianity definition of $\mathcal{N}(t_{2})$ is very different from the definition of $D_{NM}(t_{2},t_{1},t_{0})$ because $\mathcal{N}(t_{2})$ depends only on $t_{2}$, but $D_{NM}(t_{2},t_{1},t_{0})$ depends on $t_{0}$, $t_{1}$ and $t_{2}$, i.e., the dynamics in  $D_{NM}(t_{2},t_{1},t_{0})$ is influenced by intervention at time $t_1$
while the dynamics in $\mathcal{N}(t_{2})$ is not.
For a specific $t_{2}$, we can also define the integrated non-Markovianity $\mathcal{N}_{D}(t_{2})$ by 
\begin{equation}
\label{eq:NM}
\mathcal{N}_{D}(t_{2})=\int^{t_{2}}_{0}  dt_{1}  D_{NM}(t_{2},t_{1},0).
\end{equation}

The minimization for Eq.~(\ref{eq:non-intNM}) is not easy to perform. If we choose the distance measure $\mathcal{D}$ as quantum relative entropy $\mathcal{R}$ and normalize the matrix form of the process tenser, 
the closest Markovian process is straightforwardly found by discarding the correlations \cite{KModi,KModi2}.
However, the quantum relative entropy may diverge in some cases.
For example,
the quantum relative entropy of density matrix $\rho_1$ with respect to density matrix $\rho_2$,
$\mathcal{R}(\rho_1,\rho_2)={\rm tr}\rho_1({\rm log}\rho_1-{\rm log}\rho_2)$, is well-defined only for any pair of positive semi-definite matrices for which  ker($\rho_2$) $\subset$ ker($\rho_1$)  \cite{ReEntropy}. Here, ker($\rho$)$= \{ v \in \mathcal{H} | \rho\, v = 0\}$ is called the kernel (or null space) of $\rho: \mathcal{H} \rightarrow \mathcal{H}$ and is the set of state vectors $v$ in the vector space $\mathcal{H}$ satisfying $\rho\, v =0$.
Therefore, we present an algorithm for finding $D_{NM}$ by choosing $\mathcal{D}$ to be a convex function in next subsection.

\subsection{Optimization algorithm for finding $D_{\rm NM}$}
\label{NMopti}
We describe here an algorithm to minimize Eq.~({\ref{eq:non-intNM}). We restrict the distance measure $\mathcal{D}$ to be a convex function (e.g., trace distance). 
The convex optimization is an efficient method to find the global minimum \cite{cvxbook, cvx}. 
Unfortunately, it can not be directly implemented to Eq.~(\ref{eq:non-intNM}) because the set  $\{  T_{\rm Markov}=T_{1} \otimes T_{0} \otimes \rho^{S} \}$  is non-convex, where $\rho^{S}$ is an arbitrary system state, $T_{1}$ and $T_{0}$ are arbitrary trace preserving and completely positive (TPCP) maps which send a system state to another system state. However, the individual sets of $T_{1}$, $T_{0}$ and $\rho^{S}$ are all convex.

Therefor, we split the minimization procedure 
\begin{equation}
\min_{T_{1} \otimes T_{0} \otimes \rho^{S}}  \mathcal{D} \left(T , T_{1} \otimes T_{0} \otimes \rho^{S} \right)
\label{NMtot}
\end{equation}
of the two-time-step process into three steps: with $T_{1}$ and $T_{0}$ fixed, varying $\rho^{S}$ to minimize
\begin{equation}
\min_{\rho^{S}}  \mathcal{D} \left(T , T_{1} \otimes T_{0} \otimes \rho^{S} \right),
\label{NM0}
\end{equation}
with $T_{1}$ and $\rho^{S}$ fixed, varying $T_{0}$  to minimize
\begin{equation}
\min_{T_{0}}  \mathcal{D} \left(T , T_{1} \otimes T_{0} \otimes \rho^{S} \right)
\label{NM1}
\end{equation}
and with $T_{0}$ and $\rho^{S}$ fixed, varying $T_{1}$  to minimize  
\begin{equation}
\min_{T_{1}} \mathcal{D} \left(T , T_{1} \otimes T_{0} \otimes \rho^{S} \right).
\label{NM2}
\end{equation}

Since the varying feasible sets and the objective function $\mathcal{D}$ in
Eq.~(\ref{NM0}), Eq.~(\ref{NM1}) and Eq.~(\ref{NM2}) are all convex, and the convex optimization algorithm can be employed. If we could test every posible $T_{1}$ and $T_{0}$ in Eq.~(\ref{NM0}), the minimum of Eq.~(\ref{NMtot}) would be found. However, a much more efficient way to perform the optimization is given by the following steps. 

\begin{steps}
  \item Guess trials $ T_{1}$ and $T_{0}$.
  \item Solve Eq.~(\ref{NM0}) by the convex optimization algorithm to obtain the minimum value and optimal $\rho^{S}$. Here, $ T_{1}$ and $T_{0}$ are given by the previous step.
	\item Solve Eq.~(\ref{NM1}) by the convex optimization algorithm to obtain the minimum value and optimal $T_{0}$. Here, $ T_{1}$ and $\rho^{S}$ are given by the previous step.
  \item Solve Eq.~(\ref{NM2}) by convex optimization algorithm to obtain the minimum value and optimal $T_{1}$. Here, $ T_{0}$ and $\rho^{S}$ are given by the previous step.
	\item Repeat  Step 2, Step 3 and Step 4  until $T_{1}$, $ T_{0}$ and $\rho^{S}$ are all invariant at the end of each loop.
	\item Repeat Step 1 to  step 5  many times to obtain the optimized distance values.
	\item Take the minimum value of the optimized distance values.
\end{steps}

One should notice that $T_{1}$ and $T_{0}$ are not only positive but also restricted by the trace preserving condition Eq.~(\ref{trace preserving}).

The optimization algorithm for the two-time-step process with an initial factorized system and environment state is simpler. In this case, the process tensor given by Eq.~(\ref{eq41}) indicates that one can fix the initial system state
$\rho^{S}_{[0]}$ in the optimization procedure.
Therefore, the corresponding modified optimization algorithm is  
just to fix $\rho^{S}=\rho^{S}_{[0]}$, ignore Step~2 and replace Step 1 with Step $1'$: Guess a trial $ T_{1}$.

\section{No IBTRES for a  two-time-step quantum process}
In this section, we present an operational definition for no IBTRES and a distance measure for IBTRES, and describe an optimization algorithm for calculating the distance measure. We will also discuss the relation of our IBTRES with other studies of information backflow in the recent literature \cite{KModiNew,NMLocal}.

\label{secIBTRES}

\begin{figure}
\includegraphics[width=0.9\columnwidth]{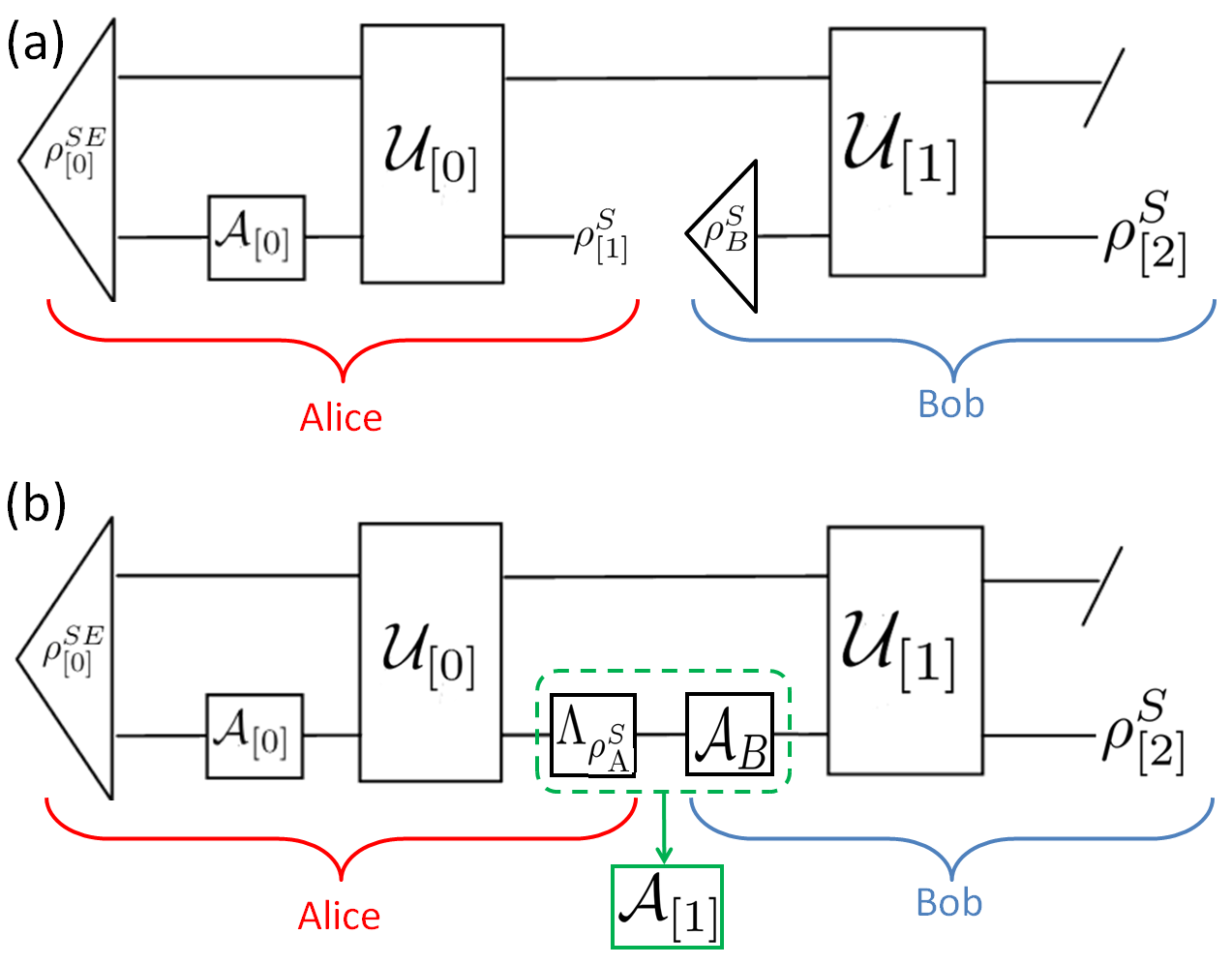}
\caption{Schematic illustrations of the Schematical illustration of a two-time-step process.
(a) Operational definition to check whether a two-time-step process has IBTRES.
(b) Equivalent circuit for (a) achieved by setting the quantum operation $\mathcal{A}_{[1]}=\mathcal{A}_{B}\circ \Lambda_{\rho^{S}_{A}}$  in Fig.~\ref{STW}(c).}
\label{Operation}
\end{figure}

\subsection{Operational definition of no IBTRES} 
The operational description of IBTRES
in a two-time-step process is defined by the following experiment.
 As shown in Fig.~\ref{Operation}(a), Alice wants to send information to Bob through an open quantum system by the following method:

\begin{steps}
  \item  Alice applies an operation $\mathcal{A}_{[0]}$ on an initial system-environment state $\rho^{S_{A}E}_{[0]}$ and sends  $\mathcal{A}_{[0]}[\rho^{S_{A}E}_{[0]}]$ to the system-environment unitary evolution $\mathcal{U}_{[0]}$.
  \item Alice keeps the system part to herself and leaves the output of environment part of Step 1 ($\rho^{E}_{[1]}$) to Bob. 
  \item  Bob produces a new arbitrary state $\rho^{S_{B}}_{B}$ for a new system with dimension dim$(\rho^{S}_{[0]})=$dim$(\rho^{S_B}_{B})$ to interact with the environment $\rho^{E}_{[1]}$ with the same system-environment Hamiltonian as the original system, and then the system-environment state ($\rho^{S_{B}}_{B} \otimes \rho^{E}_{[1]}$) undergoes the system-environment unitary $\mathcal{U}_{[1]}$.
	\item Bob measures the output of the system part $\rho^{S}_{[2]}$. If varying $\mathcal{A}_{[0]}$ results in the same $\rho^{S}_{[2]}$, then the process has no IBTRES. 
\end{steps}

 In this experiment, the information, which is sent from Alice to Bob, can pass only through the reduced environment state $\rho^{E}_{[1]}$. Note that Step 3  corresponds to Eq.~(\ref{Leq}) for $n=1$. A two-time-step process has no IBTRES if and only if Alice can not signal to Bob by this experiment.    

 Next, we obtain the if-and-only-if (iff) condition for no IBTRES in terms of process tensor. For this purpose, we should choose a suitable $\mathcal{A}_{[1]}$ in Fig.~\ref{STW}(c) to represent the experimental setting shown in Fig.~\ref{Operation}(a). This can be achieved by setting $\mathcal{A}_{[1]}$ as a constant map $\Lambda_{\rho^{S}_{A}}$ composite with 
 an arbitrary
  trace preserving and completely positive (TPCP) map
 $\mathcal{A}_{B}$ as shown in Fig.~\ref{Operation}(b).
Because $\Lambda_{\rho^{S}_{A}}[\rho^{SE}_{[1]}]=\rho^{S}_{A}\otimes\rho^{E}_{[1]}$, the experiment setting in Fig.~\ref{Operation}(b) is equivalent to original experiment setting in Fig.~\ref{Operation}(a) if one replaces Step 2 by Step 2$'$ and Step 3 by Step3$'$ as follows.

  Step 2$'$: Alice applies a constant map  $\Lambda_{\rho^{S}_{A}}$  on the same original system, sends the total output state $\Lambda_{\rho^{S}_{A}}[\rho^{SE}_{[0]}]=\rho^{S}_{A}\otimes\rho^{E}_{[1]}$ to Bob, and tells Bob what $\rho^{S}_{A}$ is.

  Step 3$'$:  Bob applies
  a quantum operation $\mathcal{A}_{B}$ on the system part of the total state $\rho^{S}_{A}\otimes\rho^{E}_{[1]}$ he received from Alice and sends the resultant state $\mathcal{A}_{B}[\rho^{S}_{A}\otimes\rho^{E}_{[1]}]$ into the unitary map $\mathcal{U}_{[1]}$,
  where $\mathcal{A}_{B}[\rho^{S}_{A}] \equiv \rho^{S}_{B}$ is the state produced by Bob in Step 3.

  The operational description of no IBTRES in Fig.~\ref{Operation}(b) is more feasible in a realistic experiment \cite{KModiNew}. 
  A similar procedure to that in Fig.~\ref{Operation}(b)  is called a {\em causal break} introduced in \cite{KModi} to break the causal
link between the past and the future of the system when discussing the criterion for a quantum Markovian process. 
Generally, any operation whose output is
independent of its input constitutes a causal break, and an example of
causal break given in \cite{KModi} is to perform
a quantum measurement on the system and
  then reprepare the post-measurement system state into a known state randomly chosen from some set.

	
\newtheorem{Corollary}{Corollary}
\begin{Corollary}
\label{local effect}
A two-time-step process of Eq.~(\ref{Eq39}) has no IBTRES if and only if the process tensor in matrix form $T$ satisfies the following equation 
\begin{equation}
\label{eq:local effect}
{\rm tr}_{S_{1}}(T)= T'^{S_{2}S_{1'}}_{[1]}\otimes I^{S_{0'}} \otimes \rho^{S}_{[0]},
\end{equation}
where $T'^{S_{2}S_{1'}}_{[1]}$ is the Choi matrix of a TPCP quantum map, and 
$I^{S_{0'}}$ is the identity matrix.
\end{Corollary}

\begin{proof}
See Appendix \ref{APPnoIB}.
\end{proof}

\subsection{Distance measure  for IBTRES}
Given the process tensor in matrix form  $T$
of a two-time-step-process and a distance measure $\mathcal{D}$, 
the measure of IBTRES of the process is defined by the minimum distance between $T$ and the set $T_{\rm no-IBTRES}$ as shown in Fig.~\ref{set}: 
\begin{equation}
\label{NLE0}
 D_{\rm IBTRES}=\min_{ T_{\rm no-IB} \in    \{ T_{\rm no-IBTRES} \}}  \mathcal{D}(T,T_{\rm no-IB}).
\end{equation}
Here, $\{T_{\rm no-IBTRES} \}$ is the collection of the two-time-step processes with no IBTRES, and $T_{\rm no-IB}$ is a possible process tensor in the set $\{ T_{\rm no-IBTRES}\}$.

\subsection{Optimization algorithm for finding $D_{\rm IBTRES}$}
In this subsection, we describe an algorithm to minimize Eq.~(\ref{NLE0}) for two-time-step processes.
By using Eq.~(\ref{eq:local effect}), Eq.~(\ref{NLE0}) can be represented by the following form:
\begin{equation}
\label{OptIB}
 D_{\rm IBTRES}=\min_{ \substack{ T'^{S_{2}S_{1'}} \in  {\rm TPCP \ map} \\ \rho^{S} \in {\rm density \ matrix} \\{\rm tr}_{S_{1}}(T_{\rm no-IB}) =T'^{S_{2}S_{1'}}\otimes I^{S_{0'}} \otimes \rho^{S} } }   \mathcal{D}(T,T_{\rm no-IB}).
\end{equation}
If the distance measure $\mathcal{D}$ is convex, we can use the similar procedure in Sec.\ref{NMopti} to perform the optimization by exploiting the convex optimization algorithm. We split
the minimization procedure of  Eq.~(\ref{OptIB}) into two steps:

(1) Fix $\rho^{S}$ and vary $T_{\rm no-IB}$ to find optimal $T'^{S_{2}S_{1'}}$ to minimize
\begin{equation}
\label{OptIB1}
\min_{ \substack{ T'^{S_{2}S_{1'}} \in  {\rm TPCP \ map} \\  {\rm tr}_{S_{1}}(T_{\rm no-IB}) =T'^{S_{2}S_{1'}}\otimes I^{S_{0'}} \otimes \rho^{S} } }    \mathcal{D}(T,T_{\rm no-IB}).
\end{equation}

(2) Fix $T'^{S_{2}S_{1'}}$ and vary $T_{\rm no-IB}$ to find optimal $\rho^{S}$ to minimize 
\begin{equation}
\label{OptIB2}
\min_{ \substack{  \rho^{S} \in {\rm density \ matrix} \\ {\rm tr}_{S_{1}}(T_{\rm no-IB}) =T'^{S_{2}S_{1'}}\otimes I^{S_{0'}} \otimes \rho^{S} } }    \mathcal{D}(T,T_{\rm no-IB}).
\end{equation}

Eq.~(\ref{OptIB1}) and Eq.~(\ref{OptIB2}) are both convex and the convex optimization algorithm can be implemented. 
Similar to the algorithm in Sec.\ref{NMopti}, Eq.~(\ref{OptIB}) can be performed by the following algorithm:
\begin{steps}
  \item Guess a trial $\rho^{S}$.
  \item Solve Eq.~(\ref{OptIB1}) by the convex optimization algorithm to obtain the minimum value and optimal $T'^{S_{2}S_{1'}}$. Here, $\rho^{S}$ is given by the previous step.
	\item Solve Eq.~(\ref{OptIB2}) by the convex optimization algorithm to obtain the minimum value and optimal $\rho^{S}$. Here, $T'^{S_{2}S_{1'}}$ is given by the previous step.
	\item Repeat  Step 2 and Step 3 until $\rho^{S}$ and $T'^{S_{2}S_{1'}}$ are both invariant at the end of each loop.
	\item Repeat Step 1 to  step 4 many times to obtain the optimized distance values.
	\item Take the minimum value of the optimized distance values.
\end{steps}

The algorithm can be much simplified if the following two additional constrains are both satisfied: (a) the initial system-environment state of the process is uncorrelated or factorized, $\rho^{SE}_{[0]}=\rho^{S}_{[0]}\otimes \rho^{E}_{[0]}$, so the process tensor in matrix form $T = \tilde{T} \otimes \rho^{S}_{[0]}$ is also in an uncorrelated from of Eq.~(\ref{eq41}). (b) the distance measure $\mathcal{D}$ is not only a convex function but also contractive under trace preserving and completely positive (TPCP) maps, i.e., if $\mathcal{A}_{\rm TPCP}$ is a TPCP map, then $\mathcal{D}(\rho_{1},\rho_{2})\geq\mathcal{D}(\mathcal{A}_{\rm TPCP}[\rho_{1}],\mathcal{A}_{\rm TPCP}[\rho_{2}])$ for arbitrary density matrices $\rho_{1}$ and $\rho_{2}$. The trace distance, for example, as a distance measure satisfies condition (b). 
With these two additional constrains, one obtains 
\begin{eqnarray}
  \mathcal{D}(T,T_{\rm no-IB}) & =& \mathcal{D}(\tilde{T} \otimes \rho^{S}_{[0]},T_{\rm no-IB}) \nonumber \\ 
&\geq &\mathcal{D}( \Lambda_{\rho^{S}_{0}}[\tilde{T} \otimes \rho^{S}_{[0]}],\Lambda_{\rho^{S}_{0}}[T_{\rm no-IB}]) \nonumber \\
&=& \mathcal{D}( \tilde{T} \otimes \rho^{S}_{[0]},\tilde{T}_{\rm no-IB} \otimes \rho^{S}_{[0]}),
\label{eq55}
\end{eqnarray}
where $\Lambda_{\rho^{S}_{0}}$ is the constant map (a TPCP map) applied only on the system state space and $\tilde{T}_{\rm no-IB} = {\rm tr}_{S}(T_{\rm no-IB})$.
Equation~(\ref{OptIB}) can be simplified by the inequality of Eq.~(\ref{eq55}):
\begin{equation}
\label{OptIBfac0}
\begin{split}
 &D_{\rm IBTRES}= \min_{ \substack{ T'^{S_{2}S_{1'}} \in  {\rm TPCP \ map} \\ \rho^{S} \in {\rm density \ matrix} \\ {\rm tr}_{S_{1}}(T_{\rm no-IB}) =T'^{S_{2}S_{1'}}\otimes I^{S_{0'}} \otimes \rho^{S} } }  \mathcal{D}(T,T_{\rm no-IB}) \\
&= \min_{ \substack{ T'^{S_{2}S_{1'}} \in  {\rm TPCP \ map} \\ \rho^{S} \in {\rm density \ matrix} \\ {\rm tr}_{S_{1}}(T_{\rm no-IB}) =T'^{S_{2}S_{1'}}\otimes I^{S_{0'}} \otimes \rho^{S} } } \mathcal{D}( \tilde{T} \otimes \rho^{S}_{[0]},\tilde{T}_{\rm no-IB} \otimes \rho^{S}_{[0]}). 
\end{split}
\end{equation}
Note that the minimum value will happen at $\rho^{S}=\rho^{S}_{[0]}$. So $\rho^{S}$ is fixed, and can be remove from the minimizeation variables. Therefore Eq.~(\ref{OptIBfac0}) becomes
\begin{equation}
\label{OptIBfac}
\begin{split}
&D_{\rm IBTRES} \\
&= \min_{ \substack{ T'^{S_{2}S_{1'}} \in  {\rm TPCP \ map}  \\ {\rm tr}_{S_{1}}(\tilde{T}_{\rm no-IB}) =T'^{S_{2}S_{1'}}\otimes I^{S_{0'}} } } \mathcal{D}( \tilde{T} \otimes \rho^{S}_{[0]},\tilde{T}_{\rm no-IB} \otimes \rho^{S}_{[0]}). \\
\end{split}
\end{equation}
Equation~(\ref{OptIBfac}) can be solved by the convex optimization algorithm efficiently and the condition  
\begin{equation}
\label{ucNoIB}
{\rm tr}_{S_{1}}\tilde{T} =T'^{S_{2}S_{1'}}\otimes I^{S_{0'}}
\end{equation}
in the minimization is the iff condition for a process with no IBTRES and with an initial factorized system-environment state.

\subsection{Relation with other information backflow studies}
We notice a recent study that discusses the completely positive (CP) divisibility for a quantum process with an initial factorized system-environment state \cite{KModiNew} using the idea of no information backflow. 
The conditional non-signaling condition defined in \cite{KModiNew} for a two-time-step process (time: $r$, $s$, $t$) is the no-IBTRES process with an initial factorized system-environment state discussed here. Our description here for no-IBTRES process is, however, for a general initial correlated system-environment state.
For an $m$-time-step process, the authors in \cite{KModiNew}
define an \textsl{operational} dynamical map $\Xi_{n_{2},n_{1}}$ from time step $n_{1}$ to time step $n_{2}$ (with $n_{1}< n_{2}$) by setting the quantum operation $\mathcal{A}_{[n_{1}]}$ at time step $n_{1}$ as a constant map, i.e. $\mathcal{A}_{[n_{1}]}=\Lambda_{\rho}$ and measure the output of the system state at time step $n_{2}$.
The \textsl{operational} dynamical map $\Xi_{n_{2},n_{1}}$ can then be obtained by varying the state $\rho$ of $\Lambda_{\rho}$ and performing the process tomography for $\Xi_{n_{2},n_{1}}$.
They define that the dynamics of an $m$-time-step process is operational CP divisible if and only if $\Xi_{m,0} = \Xi_{m,n}\circ \Xi_{n,0}$. In this case, the process has no IBTRES for the two-time-step process with time steps: ($0$,$n$,$m$).
The operational CP divisibility defined in \cite{KModiNew} using a constant map 
is a necessary but not sufficient condition for a Markovian process.

We notice that another idea of no information backflow similar to the no-IBTRES effect has been proposed by L. Li et al. \cite{NMLocal}. There, the constant map is applied on the environment rather than the system (see Appendix \ref{Li}).
In contrast, the general quantum regression formula
described also in \cite{NMLocal} involves interventions
at different times on the open system itself and
is very closely related to the approach for a Markovian process
discussed here.
We show in Appendix \ref{GQRT} that
the general quantum regression formula is a sufficient but
not necessary condition for a Markovian process.
We also define the {\em extended} general quantum regression formula (see Definition \ref{postGQRTdef} in Appendix \ref{GQRT}) 
and prove that a process is Markovian if and only if it satisfies the extended general quantum regression formula.

\section{No SECE  for a  two-time-step quantum process}
In this section, we present an operational definition for no SECE and a distance measure for SECE.
\subsection{Operational definition}
We give here the iff condition of a two-time-step process that has no SECE in terms of process tensor. The conditions of 
no SECE for a single-time-step process ($M$), and the first-time-step process of a two-time-step process [see Eq.~(\ref{QChanneln}] for $n=1$) are equivalent to the process tensor condition of an initial uncorrelated or factorized system-environment state. 
Thus from Eq.~(\ref{eq41}), the corresponding process tensors 
are given by the product forms:
\begin{equation}
M=T_{[1]} \otimes \rho^{S}_{[0]}
\end{equation}
and 
\begin{equation}
T=\tilde{T} \otimes \rho^{S}_{[0]},
\end{equation}
respectively,
where, $\tilde{T} \in \mathcal{B}\left(\mathcal{H}^{S_{2}S_{1'}S_{1}S_{0'}}\right)$.
In the remaining part of this subsection, we focus on the iff condition of no SECE for the process in the second time step of the two-time-step process (see Eq.~(\ref{QChanneln}) for $n=2$). 
\newtheorem{Lemma}[theorem]{Lemma}
\begin{Lemma}
\label{Lemma1}
The second-time-step process of a two-time-step process has no SECE if and only if the process tensor in map form satisfies
\begin{equation}
\label{eq57}
\mathcal{T}[\mathcal{A}_{[1]},\mathcal{A}_{[0]}] = \mathcal{T}[\mathcal{A}_{[1]} \circ \Lambda_{\rho^{S}_{[1]}/{\rm tr}(\rho^{S}_{[1]})},\mathcal{A}_{[0]}], 
\end{equation}
or equivalently, in tensor form satisfies 
\begin{equation}
\label{Lemma1eq}
\begin{split}
& T^{i_{1'}j_{1'},i_{0'}j_{0'}}_{\ i_{2}j_{2},\  i_{1}j_{1},\ i_{0}j_{0}}\mathcal{A}_{[0];i_{0'}j_{0'}}^{\ \ \ i_{0}j_{0}}\mathcal{A}_{[1];i_{1'}j_{1'}}^{\ \ \ i_{1}j_{1}}(\delta_{x_{0'}y_{0'}}\rho^{S}_{[1];x_{0'},y_{0'}}) \\ 
 &= T^{i_{1'}j_{1'},i_{0'}j_{0'}}_{\ i_{2}j_{2},\  x_{0}y_{0},\ i_{0}j_{0}}\mathcal{A}_{[0];i_{0'}j_{0'}}^{\ \ \ i_{0}j_{0}}\mathcal{A}_{[1];i_{1'}j_{1'}}^{\ \ \ i_{1}j_{1}}\cdot \delta_{x_{0}y_{0}}\rho^{S}_{[1];i_{1},j_{1}}.
\end{split}
\end{equation}
\end{Lemma}

Lemma \ref{Lemma1} follows direct from the fact that $\Lambda_{\rho^{S}_{[1]}/{\rm tr}(\rho^{S}_{[1]})}[\rho^{SE}_{[1]}] = \rho^{S}_{[1]} \otimes \rho^{E}_{[1]}/{\rm tr}(\rho^{S}_{[1]}) = \rho^{S}_{[1]} \otimes \rho^{E}_{[1]}/{\rm tr}(\rho^{E}_{[1]})$ and Eq.~(\ref{no-SECEeq}) in Definition \ref{no-SECEdef}. 
 The factor  
$\delta_{x_{0}y_{0}}\rho^{S}_{[1];i_{1},j_{1}}$ and $(\delta_{x_{0'}y_{0'}}\rho^{S}_{[1];x_{0'},y_{0'}})$
 in Eq.~(\ref{Lemma1eq}) corresponds to the constant map 
\begin{equation}
\label{trCmap}
\Lambda_{\rho^{S}_{[1]}/{\rm tr}(\rho^{S}_{[1]})}[\rho] = {\rm tr}(\rho) \rho^{S}_{[1]}/{\rm tr}(\rho^{S}_{[1]}),
\end{equation}
where the divisor ${\rm tr}(\rho^{S}_{[1]})$  makes $\Lambda_{\rho^{S}_{[1]}/{\rm tr}(\rho^{S}_{[1]})}$  trace-preserving.
In obtaining Eq.~(\ref{Lemma1eq}),
we have multiplied the factor ${\rm tr}(\rho^{S}_{[1]})$ (i.e., $\delta_{x_{0'}y_{0'}}\rho^{S}_{[1];x_{0'},y_{0'}}$) to the both sides of  Eq.~(\ref{eq57}).

The iff condition, Eq.~(\ref{eq57}), of Lemma \ref{Lemma1} is elegantly described but hard to verify in practice as one needs to know $\rho^{S}_{[1]}$ to construct $\Lambda_{\rho^{S}_{[1]}/{\rm tr}(\rho^{S}_{[1]})}$ which violate the no-cloning theorem. We give below a verifiable iff condition that requires only the process tensor $T$ in
matrix form is known.

To this end, let us first set 
\begin{eqnarray}
  \label{eq:L}
&L = {\rm tr}_{S_{1}}(T) \\
&N ={\rm tr}_{S_{1}}(M)
\label{eq:N}    
\end{eqnarray}
where  $M = \frac{1}{n_{d}} {\rm tr}_{S_{2}{S_{1'}}} (T)$, with $n_d$ the dimension of the Hilbert space, is the single-time-step process tensor in matrix form, and $L$ is the matrix form of $\mathcal{L}_{[n]}$ in Eq.~(\ref{Leq}) for $n=1$.
Then the iff  condition in matrix form in terms of the tensor products of the matrices ($N\otimes T$) and  ($L\otimes M$) in a relabeled basis state vector can be obtain as shown in Corollary \ref{SECEprop1} below.

\begin{Corollary}
\label{SECEprop1}
The second-time-step process in a two-time-step process has no SECE if and only if 
\begin{equation}
\label{CE_matrix0}
\begin{split}
&\left( \mathcal{I} + \mathcal{S}_{30'} \circ \mathcal{S}_{2'0} \right)N \otimes T \\
&= \mathcal{S}_{32} \circ \mathcal{S}_{2'1'} \circ \left( \mathcal{I} + \mathcal{S}_{20'} \circ \mathcal{S}_{1'0} \right)L \otimes M,
\end{split}
\end{equation}
where the indexes of the basis vectors in the matrix form of Eq.~(\ref{CE_matrix0}) are written as 
$\ket{i_{3}}_{3}\bra{j_{3}}\otimes\ket{i_{2'}}_{2'}\bra{j_{2'}}\otimes\ket{i_{2}}_{2}\bra{j_{2}}\otimes\ket{i_{1'}}_{1'}\bra{j_{1'}}\otimes\ket{i_{1}}_{1}\bra{j_{1}}\otimes\ket{i_{0'}}_{0'}\bra{j_{0'}}\otimes\ket{i_{0}}_{0}\bra{j_{0}}$,  $\mathcal{I}$ is the identity map and
$\mathcal{S}_{mn}$ is the SWAP operation which swaps the basis state indexes between $m$ and $n$, e.g., 
$\mathcal{S}_{2'1'}[\cdots \otimes\ket{{\color{red}i}}_{2'}\bra{{\color{red}j}}\otimes \cdots \ket{{\color{blue}k}}_{1'}\bra{{\color{blue}l}}\otimes \cdots ] = \cdots \otimes\ket{{\color{blue}k}}_{2'}\bra{{\color{blue}l}}\otimes \cdots \ket{{\color{red}i}}_{1'}\bra{{\color{red}j}}\otimes \cdots$.
 
\end{Corollary}

\begin{proof}
See Appendix \ref{AppSECEprop1}.
\end{proof}

\begin{Corollary}
\label{SECEco}
A two-time-step process has no SECE if and only if 
\begin{equation}
\label{P1Markov}
T=\tilde{T} \otimes \rho^{S}_{[0]}
\end{equation}
and 
\begin{equation}
\label{CE_matrix}
(\mathcal{I} + \mathcal{S}_{0'0})[\tilde{T} \otimes I ] = (\mathcal{I} + \mathcal{S}_{0'0}) \circ \mathcal{S}_{10'} [\tilde{L} \otimes \tilde{M}],
\end{equation}
where $\mathcal{I}$ is the identy map,
$\tilde{T} = {\rm tr}_{S_{0}}T$, $\tilde{M} = {\rm tr}_{S_{0}}M$, $\tilde{L} = {\rm tr}_{S_{0}}L$, and
the SWAP operaiton $\mathcal{S}_{mn}$ is defined in Corollary \ref{SECEprop1}.
\end{Corollary}
\begin{proof}
See Appendix \ref{AppSECEcoro1}.
\end{proof}
Note that the map $(\mathcal{I} + \mathcal{S}_{0'0})$ in Eq.~(\ref{CE_matrix}) is non-invertible, and thus it can not be removed from the both sides of the equation.
As on can obtain process tensor  $T$ experimentally, Eq.~(\ref{P1Markov}) and Eq.~(\ref{CE_matrix}) in Corollary \ref{SECEco} give the operational description of the no SECE since $\tilde{T}$, $\tilde{M}$ and $\tilde{L}$ can be calculated form the experimentally obtained $T$.

\subsection{Distance measure for SECE}
As before, one can define a measure of SECE, $D_{\rm SECE}$, as
\begin{equation}
\label{minDSECE}
 D_{\rm SECE}=\min_{ T_{\rm no} \in    \{ T_{\rm no-SECE} \}}  \mathcal{D}(T,T_{\rm no}).
\end{equation}
However, 
the constrain $T_{\rm no} \in    \{ T_{\rm no-SECE} \}$  given by Eq.~(\ref{CE_matrix0})
is hard to calculate, so we do not give an algorithm for Eq.~(\ref{minDSECE}) here.
A simple but not so accurate way for SECE is to use the 
distance difference between the two sides of Eq.~(\ref{CE_matrix}).

\section{Solely SECE for a  two-time-step quantum process} 
From the geometry point of view of Fig \ref{set},
we can define the distance of the solely SECE. 
We describe below the distance measure for solely SECE, and an optimization algorithm for calculating the distance measure.

\subsection{Distance measure  for Solely SECE}
\begin{mydef} 
 The distance of the solely SECE of a quantum process T is defined by ${\mathcal{D}(T^{{\rm min}}_{{\rm Markov}},T^{{\rm min}}_{{\rm no-IBTRES}})}$. Here,  $\mathcal{D}$ is some distance measure, $T^{{\rm min}}_{{\rm Markov}}$ and $T^{{\rm min}}_{{\rm no-IBTRES}}$ are the process tensors in the Markovian process set and the no-IBTRES process set, respectively (see Fig.~\ref{set}), which have the minimum distance with the process tensor $T$.
\end{mydef}
 If  $T^{{\rm min}}_{{\rm Markov}}$ and $T^{{\rm min}}_{{\rm no-IBTRES}}$ are not unique, one should take the minimum value between them: 
\begin{gather}
\label{DefCE}
  D_{{\rm solely-SECE}} 
	=\min_{\substack{T_{M}\in \left\{T^{{\rm min}}_{{\rm Markov}}\right\} \\
	T_{{\rm no-IB}}\in \left\{T^{{\rm min}}_{{\rm no-IBTRES}}\right\}}} \mathcal{D}(T_{M},T_{{\rm no-IB}}).
\end{gather}

\begin{Corollary}
\label{coro2}
A Markovian process has no solely-SECE.
\end{Corollary}
\begin{proof}
From Proposition \ref{mainProp}, a Markovian process has no IBTRES. Thus the statement in Corollary \ref{coro2} is obvious 
because the set of solely-SECE is a subset of no IBTRES.
Alternatively, one can argue it by the triangle inequality.
Let $T$ be the process tensor in matrix form of the process we consider.
The non-Markovianity, no IBTRES  and solely SECE can be characterized by $\mathcal{D}(T,T^{{\rm min}}_{{\rm Markov}})$, $\mathcal{D}(T,T^{{\rm min}}_{{\rm no-IBTRES}})$ and $\mathcal{D}(T^{{\rm min}}_{{\rm Markov}},T^{{\rm min}}_{{\rm no-IBTRES}})$, respectively, which satisfy
the triangle inequality: 
\begin{eqnarray}
\label{eq65}
		&&| \mathcal{D}(T,T^{{\rm min}}_{{\rm Markov}})-\mathcal{D}(T,T^{{\rm min}}_{{\rm no-IBTRES}})| \nonumber\\
		&\leq& \mathcal{D}(T^{{\rm min}}_{{\rm Markov}},T^{{\rm min}}_{{\rm no-IBTRES}}) \nonumber\\
		&\leq& | \mathcal{D}(T,T^{{\rm min}}_{{\rm Markov}})+\mathcal{D}(T,T^{{\rm min}}_{{\rm no-IBTRES}})|.
\end{eqnarray}
Because $T$ is Markovian, $\mathcal{D}(T,T^{{\rm min}}_{{\rm Markov}})=0$. 
Furthermore, a Markovian process has no IBTRES, so  $\mathcal{D}(T,T_{{\rm no-IBTRES}})=0$. 
As a result, one obtains form Eq.~(\ref{eq65}), $\mathcal{D}(T^{{\rm min}}_{{\rm Markov}},T^{{\rm min}}_{{\rm no-IBTRES}})=0$.
\end{proof}
The triangle inequality is useful for processes with $D_{{\rm no-IBTRES}} \approx 0$. in this case, the solely SECE can be found by $D_{{\rm solely-SECE}} \approx  D_{{\rm NM}}$. 

\subsection{Optimization Algorithm for finding $D_{\rm solely-SECE}$}
The measure for the degrees of solely SECE is defined by Eq.~(\ref{DefCE}). A naive method is using all possible process tensers $T_{M} \in  \left\{ T^{min}_{\rm Markov}  \right\}$  and $T_{\rm no-IB} \in  \left\{ T^{min}_{\rm no-IBTRES}  \right\}$ to find the minimum distance  by calculating the distance between every pair of the process tensors in these two sets. However, one can make the optimization easier by splitting Eq.(\ref{DefCE}) into 
\begin{gather}
\min_{\substack{T_{M}\in \left\{T^{{\rm min}}_{{\rm Markov}}\right\} \\
	}} \mathcal{D}(T_{M},T_{{\rm no-IB}}),
\end{gather}
and 
\begin{gather}
\label{DefCE11}
\min_{\substack{ \\
	T_{{\rm no-IB}}\in \left\{T^{{\rm min}}_{{\rm no-IBTRES}}\right\}}} \mathcal{D}(T_{M},T_{{\rm no-IB}}).
  \end{gather}
  If the initial system-environment state is uncorrelated, $\rho^{SE}_{[0]}=\rho^{S}_{[0]}\otimes\rho^{E}_{[0]}$ and the distance measure $\mathcal{D}$ is not only a convex function but also contractive under TPCP maps, then the convex optimization algorithm can be implement to Eq.~(\ref{DefCE11}).
By using an equation similar to Eq.~(\ref{eq55}) with $T \rightarrow T_{M}$, Eq.~(\ref{DefCE11}) becomes
\begin{gather}
\min_{\substack{
	\tilde{T}_{{\rm no-IB}}\otimes \rho^{S}_{[0]}\in \left\{T^{{\rm min}}_{{\rm no-IBTRES}}\right\}}}
	\mathcal{D}(\tilde{T}_{M}\otimes \rho^{S}_{[0]},\tilde{T}_{{\rm no-IB}}\otimes \rho^{S}_{[0]}),
	\label{eq70}
\end{gather}
where $\tilde{T}_{M}$ and $\tilde{T}_{\rm no-IB}$ are the reduced process tensors in matrix form. 
The minimization constraint $T_{\rm no-IB} \in \{T^{{\rm min}}_{{\rm no-IBTRES}}\}$  in Eq.~(\ref{eq70}) can be regarded equivalently as
$T \in \{T^{{\rm}}_{{\rm no-IBTRES}}\}$ and $\mathcal{D}(T,T_{\rm no-IB})=D_{\rm IBTRES}$.
Moreover, one can replace $\mathcal{D}(T,T_{\rm no-IB})=D_{\rm IBTRES}$ with 
$\mathcal{D}(T,T_{\rm no-IB}) \leq D_{\rm IBTRES}$ as the minimum value is just $D_{\rm IBTRES}$. 
This inequality relation $\mathcal{D}(T,T_{\rm no-IB}) \leq D_{\rm IBTRES}$ , however, 
makes the minimization constrain formally convex.
Therefor, Eq.~(\ref{eq70}) can be written as
\begin{gather}
		\min_{\substack{
	\tilde{T}_{{\rm no-IB}}\otimes \rho^{S}_{[0]}\in \left\{T^{{\rm }}_{{\rm no-IBTRES}}\right\} \\
	\mathcal{D}(T,\tilde{T}_{{\rm no-IB}}\otimes \rho^{S}_{[0]})\leq D_{\rm IBTRES} }} \mathcal{D}(\tilde{T}_{M}\otimes \rho^{S}_{[0]},\tilde{T}_{{\rm no-IB}}\otimes \rho^{S}_{[0]}),
	\label{eq71}
      \end{gather}
and one can perform the minimization using the convex optimization algorithm.       
Similar to the arguments for obtaining Eq.~(\ref{OptIBfac}) from Eq.~(\ref{OptIBfac0}), one can obtain from Eq.~(\ref{eq71}) the following equation:
\begin{gather}
	\min_{\substack{
	{\rm tr}_{S_{1}}\tilde{T} =T'^{S_{2}S_{1'}}\otimes I^{S_{0'}} \\
	\mathcal{D}(T,\tilde{T}_{{\rm no-IB}}\otimes \rho^{S}_{[0]})\leq D_{\rm IBTRES} }} \mathcal{D}(\tilde{T}_{M}\otimes \rho^{S}_{[0]},\tilde{T}_{{\rm no-IB}}\otimes \rho^{S}_{[0]}).
\end{gather}

The algorithm for obtaining the $D_{\rm solely-SECE}$ is then given by following steps.
\begin{steps}
  \item Find $\{T^{{\rm min}}_{{\rm Markov}}\}$ using the optimization algorithm for non-Markoviantiy $D_{\rm NM}$.
	\item Find $D_{\rm IBTRES}$ using the optimization algorithm for $D_{\rm IBTRES}$.
	\item Solve Eq.~(\ref{eq71}) using the convex optimization algorithm for every $T_{M} \in \{T^{{\rm min}}_{{\rm Markov}}\}$ found in Step 1, and record the result.
	\item Take the minimum value of the optimized distance values recorded in Step 3.
\end{steps}
Note that in this algorithm, $D_{\rm  NM}$, $D_{\rm  IBTRES}$ and $D_{\rm solely-SECE}$ are obtained by Step 1, Step 2 and Step 4, respectively.

\section{Difference between IBTRES  and SECE}
\label{difference}
There is a distinct difference between IBTRES and SECE.
IBTRES induces non-Markovian effects by sending the information of quantum operations $\mathcal{A}_{[0]}$ (for the case of an initial factorized system-environment state, it is the initial system state $\mathcal{A}_{[0]}[\rho^{S}_{[0]}]$ information ) to final state. However, SECE can induce non-Markovian effects without sending any previous system information.

\begin{prop}
SECE can induce non-Markovian effects without sending any information of $\mathcal{A}_{[0]}$  .
\label{prop:CorrNoInf}
\end{prop}
\begin{proof}
Here,  we prove the statement by constructing a valid example (see Fig.~\ref{CorrNoInf0}). Let the initial total state be a tripartite state $\rho^{S}\otimes\ket{0}_{E_{1}}\bra{0}\otimes\ket{0}_{E_{2}}\bra{0}$, where, $\rho^{S}=\mathcal{A}_{[0]}[\rho^{S}_{[0]}]$ is an arbitrary system state that the experimenter produces initially, and  $\ket{0}_{E_{1}}\bra{0}\otimes\ket{0}_{E_{2}}\bra{0}$ is the bipartite environment state. Let $\mathcal{U}_{[0]}$ be the total unitary applied on  the total state and let
\begin{equation}
{\rm tr}_{E_{2}}\mathcal{U}_{[0]}[\rho^{S}\otimes\ket{0}_{E_{1}}\bra{0}\otimes\ket{0}_{E_{2}}\bra{0}] \equiv \Lambda_{\rho^{SE_{1}}_{\rm const}}[\rho^{S}] = \rho^{SE_{1}}_{\rm const},
\end{equation}
where $\Lambda_{\rho^{SE_{1}}_{\rm const}}$ is a constant map which maps an arbitrary system state to a fixed biparties system-environment state $\rho^{SE_{1}}_{\rm const}$. The  state $\rho^{SE_{1}}_{\rm const}$ takes no information about the initial state $\rho^{S}$ because the constant map $\Lambda_{\rho^{SE_{1}}_{\rm const}}$ erase it. In a real experiment, this map can be produced by the long time evolution of a dissipative system. Next, the opeartion $A_{[1]}$ is applied on the state $\rho^{SE_{1}}_{\rm const}$, and after the bipartie unitary $\mathcal{U}_{[1]}$, one obtains the final state $\rho^{S}_{[2]}$ as 

\begin{equation}
\label{corNoInfo}
{\rm tr}_{E_{1}}\mathcal{U}_{[1]}A_{[1]}[\rho^{SE_{1}}_{\rm const}]  = \rho^{S}_{[2]}
\end{equation}
Note that Eq.~(\ref{corNoInfo}) can be non-Markovian if SECE exists. This process has no IBTRES because the constant map erase it, so the non-Markovian effect is produced only by SECE. In other words, this non-Markovian process has only SECE but takes no the initial state information to the final state.

\end{proof}

\begin{figure}
\includegraphics[width=0.9\columnwidth]{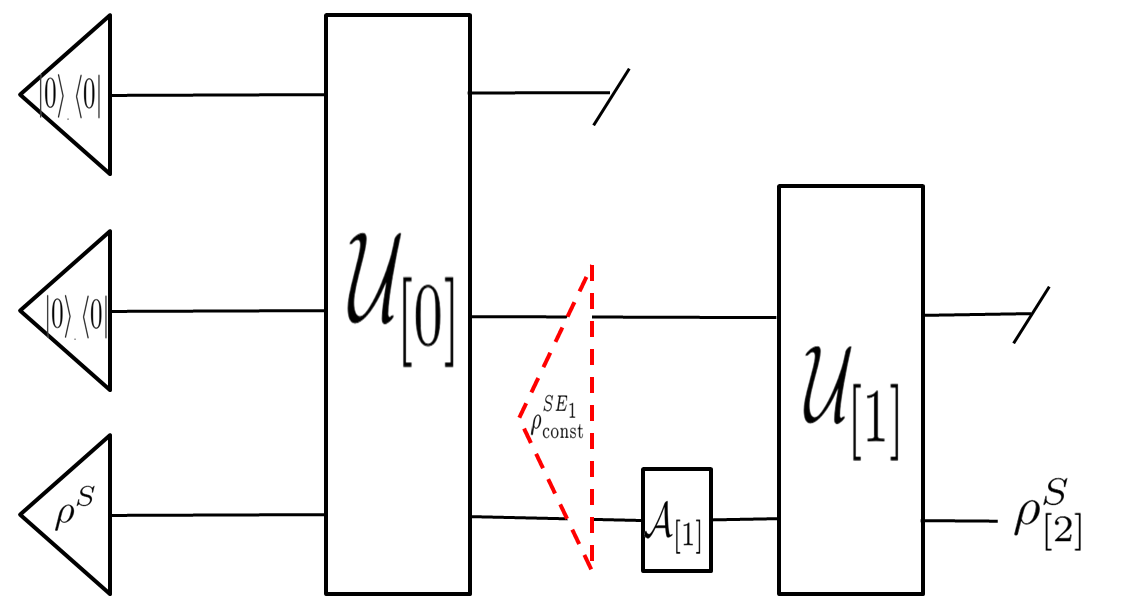}
\caption{Schematic illustration of an example for Proposition~\ref{prop:CorrNoInf}. }
\label{CorrNoInf0}
\end{figure}

\section{Conclusion}
We have characterized a Markovian process of an open quantum system as a quantum
process that has no IBTRES and no SECE.
In addition to showing that  
a process is Markovian if and only if it has no IBTRES and no SECE, and 
that the process tensor of a  two-time-step
Markovian process must be in a product form of Eq.~(\ref{MarkovForm}),
we go one step further than \cite{KModi,KModi2}
to present how to construct the maps $\mathcal{T}_{[0]}$ and $\mathcal{T}_{[1]}$ [see Eq.~(\ref{T0eq}), Eq.~(\ref{T1eq}) and Eq.~(\ref{Leq})].
We have also provided the operational definitions of no IBTRES and no SECE and derived necessary and sufficient (iff) conditions of no IBTRES and no SECE in terms of process tenser representation for a two-time-step process. We have shown moreover that SECE can produce non-Markovian effect without carrying any initial system state information but IBTRES can not.
The distance measures and algorithms for calculating non-Markovianity, IBTRES and solely-SECE have also been explicitly presented.
We will present the results of applying the qualitative definitions and quantitative distance measures discussed here to various open quantum bit systems or spin-boson models in a separate paper. 
Our study will help experimentalists in characterizing Markovian and non-Markovian processes and noises for building near-term quantum technologies. 

\begin{acknowledgments}
Y.Y.H. and H.S.G. thank K. Modi for useful discussions in the early stage of this work. H.S.G. acknowledges support from the the Ministry of Science and Technology
of Taiwan under Grants No.~MOST 106-2112-M-002-013-MY3, No.~MOST
107-2622-8-002-018 and No.~MOST 107-2627-E-002-002, and from the National
Taiwan University under Grants No.~NTU-CC-107L892902 and No.~NTU-CC-108L893202.
Z.Y.S. and H.S.G. acknowledge support from the thematic group program of the National Center for Theoretical
Sciences, Taiwan.
\end{acknowledgments}

\begin{appendix}

\section{Trace-preserving map}
\label{HPTP}
A map $\mathcal{A}$ defined in Eq.~(\ref{eq3}) is trace-preserving if and only if  $\sum_{i_{1}}A^{i_{0}j_{0}}_{i_{1}i_{1}}=\delta_{i_{0}j_{0}}$,
or equivalently,
\begin{equation}
{\rm tr}_{\rm out}(A) = I_{\rm in},
\end{equation}
where ${\rm tr}_{\rm out}$ denotes a trace over the output Hilbert space of the map $\mathcal{A}$ and $I_{\rm in}$ is the identity matrix on the input space.
\begin{proof}
By omitting ``in'' and ``out'' subscripts, one can write Eq.~(\ref{eq3}) as
\begin{equation}
\mathcal{A}[\ket{i_{0}}\bra{j_{0}}]=\sum_{i_{1}j_{1}}A^{i_{0}j_{0}}_{i_{1}j_{1}}\ket{i_{1}}\bra{j_{1}}.
\end{equation}
Letting $\rho=\sum_{i_{0}j_{0}}\rho_{i_{0}j_{0}}\ket{i_{0}}\bra{j_{0}}$, one has
\begin{equation}
\begin{aligned}
\mathcal{A}[\sum_{i_{0}j_{0}}\rho_{i_{0}j_{0}}\ket{i_{0}}\bra{j_{0}}] &=\sum_{i_{0}j_{0}}\rho_{i_{0}j_{0}}\mathcal{A}[\ket{i_{0}}\bra{j_{0}}] \\
&=\sum_{i_{0}j_{0}}\rho_{i_{0}j_{0}}\sum_{i_{1}j_{1}}A^{i_{0}j_{0}}_{i_{1}j_{1}}\ket{i_{1}}\bra{j_{1}}.  
\end{aligned} 
\end{equation}
For $i_{0}=j_{0}$, one can, without loss of generality, set $\rho=\ket{0}\bra{0}$ such that
$\mathcal{A}[\ket{0}\bra{0}]=\sum_{i_{1}j_{1}}A^{00}_{i_{1}j_{1}}\ket{i_{1}}\bra{j_{1}}$. 
Applying the trace preserving condition, one obtains ${\rm tr}(\ket{0}\bra{0})=1={\rm tr}(\mathcal{A}[\ket{0}\bra{0}])=\sum_{i_{1}}A^{00}_{i_{1}i_{1}}$.
In general, if one sets $\rho=\ket{k}\bra{k}$, then one can prove $\sum_{i_{1}}A^{kk}_{i_{1}i_{1}}=1  ,\forall k$.

For $i_{0}\neq j_{0}$, let us construct a density matrix with four elements labeled by two indexes $i_{0}$ and $j_{0}$. One can, without loss of generality, set $i_{0}=0$, $j_{0}=1$ and $\rho=\frac{1}{2}(\ket{0}\bra{0}+\ket{1}\bra{1})+\rho_{01}\ket{0}\bra{1}+\rho_{10}\ket{0}\bra{1}=\frac{1}{2}(\ket{0}\bra{0}+\ket{1}\bra{1})+\rho_{01}\ket{0}\bra{1}+\rho^{*}_{01}\ket{1}\bra{0}$. Here $\rho_{01}$ is chosen such that $\rho$ is a density matrix satisfying
$\rho=\rho^{\dagger}$, $\rho \ge 0$ and
${\rm tr}(\rho^2) \le 1$.
Applying the trace preserving condition,
\begin{equation}
\begin{aligned}
  {\rm tr}\rho
  &=1 \nonumber\\
  &={\rm tr}(\mathcal{A}[\frac{1}{2}(\ket{0}\bra{0}+\ket{1}\bra{1})+\rho_{01}\ket{0}\bra{1}+\rho^{*}_{01}\ket{1}\bra{0}]) \\
  &= \frac{1}{2} (\sum A^{00}_{i_{1}i_{1}}+ \sum A^{11}_{i_{1}i_{1}})+\sum \rho_{01} A^{01}_{i_{1}i_{1}}+\sum \rho^{*}_{01} A^{10}_{i_{1}i_{1}}.
\end{aligned}
\end{equation}
Because we have shown $\sum_{i_{1}}A^{00}_{i_{1}i_{1}}=\sum_{i_{1}}A^{11}_{i_{1}i_{1}}=1$, 
\begin{equation}
\label{eq:Trdiag}
\rho_{01}\sum_{i_{1}}A^{01}_{i_{1}i_{1}}+\rho^{*}_{01}\sum_{i_{1}}A^{10}_{i_{1}i_{1}}=0. 
\end{equation}
Since $\rho_{01}$ and $\rho^{*}_{01}$ are linearly independent, one obtains $\sum_{i_{1}}A^{01}_{i_{1}i_{1}} = \sum_{i_{1}}A^{10}_{i_{1}i_{1}}=0 $

Combining the results for the case of $i_{0}=j_{0}$ with the case  $i_{0} \neq j_{0}$, one obtains for general indexes $i_{0}$ and $j_{0}$:
\begin{equation}
\label{trace preserving}
\sum_{i_{1}}A^{i_{0}j_{0}}_{i_{1}i_{1}}=\delta_{i_{0}j_{0}}.
\end{equation}
\end{proof}

\section{Tensor representation of a constant map}
\label{constant_rep}
In this Appendix, we show that a map  maps $\forall \rho \in \mathcal{B}(\mathcal{H}^{S_{0}})$  with ${\rm tr}_{S_{0}}(\rho)=1$ to a constant density matrix $\rho_{\rm const} \in \mathcal{B}(\mathcal{H}^{S_{1}})$  with ${\rm tr}_{S_{1}}(\rho_{\rm const})=1$, i.e., satisfies Eq.~(\ref{Cmap}), if and only if
\begin{equation}
\label{CMapForm}
\Lambda_{\rho_{\rm const}}[\ket{i_{0}}\bra{j_{0}}]=\delta_{i_{0}j_{0}}\sum_{i_{1}j_{1}}\rho_{{\rm const};i_{1},j_{1}}\ket{i_{1}}\bra{j_{1}},
\end{equation}
where $\{|i_0\rangle\}$ and $\{|i_i\rangle\}$
denote orthonormal bases of the Hilbert spaces $(\mathcal{H}^{S_{0}})$ and $(\mathcal{H}^{S_{1}})$, respectively.


\begin{proof}
  Without loss of generality, let us set $\Lambda_{\rho_{\rm const}}[\ket{i_{0}}\bra{j_{0}}]=\sum_{i_{1}j_{1}}\lambda^{i_{0}j_{0}}_{i_{1}j_{1}}\ket{i_{1}}\bra{j_{1}}$. Then the result of the constant map acting on an arbitrary density matrix
  $\rho=\sum_{i_{0}j_{0}}\rho_{i_{0}j_{0}}\ket{i_{0}}\bra{j_{0}} \in \mathcal{B}(\mathcal{H}^{S_{0}}) $
becomes
\begin{equation}
\label{eqB2}
\Lambda_{\rho_{\rm const}}[\sum_{i_{0}j_{0}}\rho_{i_{0}j_{0}}\ket{i_{0}}\bra{j_{0}}]
  =\sum_{i_{1}j_{1}i_{0}j_{0}}\lambda^{i_{0}j_{0}}_{i_{1}j_{1}}\rho_{i_{0}j_{0}}\ket{i_{1}}\bra{j_{1}}.
\end{equation}
Similar to the proof given in Appendix \ref{HPTP}, we set, without loss of generality, $\rho = \rho_{00}\ket{0}\bra{0}+\rho_{11}\ket{1}\bra{1} + \rho_{01}\ket{0}\bra{1} +  \rho^{*}_{01}\ket{1}\bra{0}$. Then from Eq.~(\ref{eqB2}), one obtains
\begin{eqnarray}
\label{CMde}
  \Lambda_{\rho_{\rm const}}[\rho]&=& \rho_{00}\lambda^{00}_{i_{1}j_{1}}\ket{i_{1}}\bra{j_{1}}+\rho_{11}\lambda^{11}_{i_{1}j_{1}}\ket{i_{1}}\bra{j_{1}} \nonumber\\
 &&+ \rho_{01}\lambda^{01}_{i_{1}j_{1}}\ket{i_{1}}\bra{j_{1}} +  \rho^{*}_{01}\lambda^{10}_{i_{1}j_{1}}\ket{i_{1}}\bra{j_{1}}.
\end{eqnarray}
Using the property that ${\rm tr}_{S_{0}}(\rho)$ equals to a finite constant $c$, i.e., $ \rho_{00}+\rho_{11} = c$,
and then taking the derivative of $\partial/(\partial \rho_{00})$  on the both sides of Eq.~(\ref{CMde}), one obtains $0=\lambda^{00}_{i_{1}j_{1}}-\lambda^{11}_{i_{1}j_{1}}$. The derivative of the left hand side is $0$ because the left hand side is a constant matrix. Thus, one can generalize the above relation to $\lambda^{i_{0}i_{0}}_{i_{1}j_{1}}=\lambda^{00}_{i_{1}j_{1}}$ for every $i_{0}$. 
Setting $ \rho_{01}=\alpha+i\beta$ with $\alpha$ and $\beta$ being real numbers,
and then taking the derivative of $\partial/(\partial \alpha)$ on the both sides of Eq.~(\ref{CMde}), one has $0=\lambda^{01}_{i_{1}j_{1}}+\lambda^{10}_{i_{1}j_{1}}$.
Similarly, for the derivetive of $\partial/(\partial \beta)$, one has $0=\lambda^{01}_{i_{1}j_{1}}-\lambda^{10}_{i_{1}j_{1}}$. From these results, one concludes $\lambda^{10}_{i_{1}j_{1}}=0$.
One can generalize the result for genreal indexes $i_{0}$ and $j_{0}$ and obtains that $\lambda^{i_{0}j_{0}}_{i_{1}j_{1}}=0$ for $i_{0} \neq j_{0}$.

Substituting the two results of $\lambda^{i_{0}i_{0}}_{i_{1}j_{1}}=\lambda^{{0}{0}}_{i_{1}j_{1}}$ being independent of $i_0$, and $\lambda^{i_{0}j_{0}}_{i_{1}j_{1}}=0$ for $i_{0} \neq j_{0}$
into Eq.~(\ref{eqB2}), one obtains
\begin{equation}
\label{eqLamda}
\Lambda_{\rho_{\rm const}}[\sum_{i_{0}j_{0}}\rho_{i_{0}j_{0}}\ket{i_{0}}\bra{j_{0}}]
  =\sum_{i_{1}j_{1}}\lambda^{i_{0}j_{0}}_{i_{1}j_{1}}\sum_{i_{0}j_{0}}\delta_{i_{0} j_{0}}\rho_{i_{0}j_{0}}\ket{i_{1}}\bra{j_{1}}.
\end{equation}
On the other hand, let $\rho_{\rm const}= \sum_{i_{1}j_{1}}\rho_{{\rm const} ;i_{1}j_{1}}\ket{i_{1}}\bra{j_{1}} \in \mathcal{B}(\mathcal{H}^{S_{1}})$, one has, 
from  Eq.~(\ref{Cmap}),
\begin{eqnarray}
\label{eqrhoC}
\Lambda_{\rho_{\rm const}}[\sum_{i_{0}j_{0}}\rho_{i_{0}j_{0}}\ket{i_{0}}\bra{j_{0}}]
&=&\sum_{i_{1}j_{1}}\rho_{{\rm const} ;i_{1}j_{1}}\ket{i_{1}}\bra{j_{1}}. 
\end{eqnarray}
Using the condition that the trace $\sum_{i_{0}j_{0}}\delta_{i_{0}j_{0}}{\rho_{i_{0}j_{0}}}=1$ for Eq.~(\ref{eqLamda}), and then comparing the resultant equation to  Eq.~(\ref{eqrhoC}), one can identify
$\lambda^{i_{0}j_{0}}_{i_{1}j_{1}}=\delta_{i_{0}j_{0}}\rho_{{\rm const};i_{1}j_{1}}$.
Finally, one arrives at
Eq.~(\ref{CMapForm}).
\end{proof}

We can extend directly the constant map for density matrix of Eq.~(\ref{CMapForm})
to the constant map for process tensor $T$ as:
\begin{equation}
  \label{CMapPTForm}
  \Lambda'_{T_{\rm const}}[\mathcal{A}[\rho^{S}_{\rm fixed}]]=T_{\rm const} \cdot {\rm tr}(\mathcal{A}[\rho^{S}_{\rm fixed}]).
\end{equation}
Here, $\mathcal{A}$ acts on the system $S$, $T_{\rm const}$ is a fixed process tensor and $\rho^{S}_{\rm fixed}$ is an arbitrary fixed density matrix.

\section{Constant map erasing the information of the system and correlation}
\label{ConstMapErase}
If $\Lambda_{\rho_{\rm const}}$:$\mathcal{B}\left(\mathcal{H}^{S}\right) \rightarrow \mathcal{B}\left(\mathcal{H}^{S}\right)$ is a constant map, 
then $\Lambda_{\rho_{\rm const}} [\rho^{SE}] =\rho_{\rm const} \otimes \rho^{E} , \forall \rho^{SE} \in $ $\mathcal{B}\left(\mathcal{H}^{SE}\right)$.
\begin{proof}
Let $ \rho^{SE} =\sum_{ij\alpha\beta}\rho^{SE}_{ij,\alpha\beta} \ket{i}\bra{j} \otimes \ket{\alpha}\bra{\beta} $
and $ \rho_{\rm const} =\sum_{mn}\rho_{{\rm const};{mn}} \ket{m}\bra{n} $. 
Then
\begin{eqnarray}
\label{eqC1}
\Lambda_{\rho_{\rm const}} [\rho^{SE}] 
&=&\Lambda_{\rho_{\rm const}}[\sum_{ij\alpha\beta}\rho^{SE}_{ij,\alpha\beta} \ket{i}\bra{j} \otimes \ket{\alpha}\bra{\beta}] \nonumber\\
&=&\sum_{ij\alpha\beta}\rho^{SE}_{ij,\alpha\beta} \Lambda_{\rho_{\rm const}} [\ket{i}\bra{j}] \otimes \ket{\alpha}\bra{\beta}\nonumber\\
&=& \sum_{ij\alpha\beta} \rho^{SE}_{ij,\alpha\beta} \delta_{ij}\sum_{mn}\rho_{{\rm const};mn}\ket{m}\bra{n} \otimes \ket{\alpha}\bra{\beta}\nonumber\\
&=&\sum_{mn}\rho_{{\rm const};mn}\ket{m}\bra{n} \otimes \sum_{ij\alpha\beta} \delta_{ij} \rho^{SE}_{ij,\alpha\beta} \ket{\alpha}\bra{\beta}\nonumber\\
&=&\rho_{\rm const} \otimes {\rm tr_{S}}( \rho^{SE})\nonumber\\
&=&\rho_{\rm const} \otimes \rho^{E},
\end{eqnarray}
where we have used Eq.~(\ref{CMapForm}) in the derivation for Eq.~(\ref{eqC1}).
One can see from  Eq.~(\ref{eqC1}) that a constant map can destroy the
system-environment correlation of $\rho^{SE}$ and erase the information
of the system state.
\end{proof}

\section{Proof of Corollary \ref{local effect} }
\label{APPnoIB}

\begin{proof}
We can represent the schematic illustration for a two-time-step process of Fig.~\ref{Operation}(b) as the following equation:
\begin{equation}
A^{\ \ \ mn}_{B;i_{1'}j_{1'}}\rho^{S}_{A;mn}\delta_{i_{1}j_{1}}T^{i_{1'}j_{1'},i_{0'}j_{0'}}_{i_{2}j_{2},i_{1}j_{1},i_{0}j_{0}}A^{\ \ \ i_{0}j_{0}}_{[0];i_{0'}j_{0'}}=\rho^{S}_{[2];i_{2}j_{2}}.
\end{equation}
Here, $\rho^{S}_{A;mn}\delta_{i_{1}j_{1}}$ is the constant map [see Eq.~(\ref{CMapForm})], $A^{\ \ \ mn}_{B;i_{1'}j_{1'}}$ is the quantum operation that Bob applys on the system. Setting $A^{\ \ \ mn}_{B;i_{1'}j_{1'}}\rho^{S}_{A;mn} \equiv \rho_{B;i_{1'}j_{1'}}$ to be an arbitrary system state produced by Bob, we obtain
an effective map from $\rho_{B}$ to $\rho_{[2]}$ as 
\begin{equation}
T'^{\ i_{1'}j_{1'}}_{[1]; i_{2}j_{2}}\rho_{B;i_{1'}j_{1'}}=\rho^{S}_{[2];i_{2}j_{2}},
\end{equation}
where
\begin{equation}
\label{local Choi}
\delta_{i_{1}j_{1}}T^{i_{1'}j_{1'},i_{0'}j_{0'}}_{i_{2}j_{2},i_{1}j_{1},i_{0}j_{0}}A^{\ \ \ i_{0}j_{0}}_{[0];i_{0'}j_{0'}}\equiv T'^{\ i_{1'}j_{1'}}_{[1]; i_{2}j_{2}}.
\end{equation}
The process has no IBTRES if and only if the final state $\rho^{S}_{[2]}$ does not depend on $\mathcal{A}_{[0]}$. That means $\mathcal{A}_{[0]}$ can not affect the map $\mathcal{T}'_{[1]}$. As a result, the process tensor $\delta_{i_{1}j_{1}}T^{i_{1'}j_{1'},i_{0'}j_{0'}}_{i_{2}j_{2},i_{1}j_{1},i_{0}j_{0}}$ in Eq.~(\ref{local Choi}) is a constant map, which maps an arbitrary $A^{\ \ \ i_{0}j_{0}}_{[0];i_{0'}j_{0'}}$ to a unique $T'^{\ i_{1'}j_{1'}}_{[1]; i_{2}j_{2}}$. Using the form of the constant map for a process tensor
of  Eq.~(\ref{CMapPTForm}) in tensor form, one can write
\begin{equation}
\label{eqD4}
\delta_{i_{1}j_{1}}T^{i_{1'}j_{1'},i_{0'}j_{0'}}_{i_{2}j_{2},i_{1}j_{1},i_{0}j_{0}}=\rho^{S}_{{\rm fixed};i_{0}j_{0}}\delta_{i_{0'}j_{0'}}T'^{\ i_{1'}j_{1'}}_{[1]; i_{2}j_{2}},
\end{equation}
with $\rho^{S}_{\rm fixed}$ a fixed density matrix to be determined.
Because ${\rm tr}_{S_{2}S_{1'}S_{1}S_{0'}}(T) \propto \rho^{S}_{[0]}$, performing the same trace operations on the right hand side of Eq.~(\ref{eqD4}) leads to a result proportional to $\rho^{S}_{\rm fixed}$. Therefore, one concludes  $\rho^{S}_{\rm fixed} = \rho^{S}_{[0]}$ 
Consequently,
\begin{equation}
\label{eq:NoLocalEffect}
\delta_{i_{1}j_{1}}T^{i_{1'}j_{1'},i_{0'}j_{0'}}_{i_{2}j_{2},i_{1}j_{1},i_{0}j_{0}}=\rho^{S}_{[0];i_{0}j_{0}}\delta_{i_{0'}j_{0'}}T'^{\ i_{1'}j_{1'}}_{[1]; i_{2}j_{2}}.
\end{equation}
Equation~(\ref{eq:NoLocalEffect}) written in matrix form is Eq.~(\ref{eq:local effect}) in the main text.
\end{proof}

\section{No information backflow effect described in \cite{NMLocal}}
\label{Li}
We discuss the definition of no information backflow defined in \cite{NMLocal} and compare their definition to ours here. Below is the definition of no information backflow in \cite{NMLocal}:
\begin{mydef}
 A process with an initial environment state $\rho^{E}_{[0]}\in \mathcal{B}\left(\mathcal{H}^{E}\right)$ exhibits no information backflow if and only if for all times $t_1$ of $t_{2}\geq t_{1}\geq t_{0}$, there exists a constant map $\Lambda_{\rho'^{E}(t_{1})}$: $\mathcal{B}\left(\mathcal{H}^{E}\right) \rightarrow \mathcal{B}\left(\mathcal{H}^{E}\right)$ such that
\begin{eqnarray}
\label{eqE1}
&&\mathcal{U}_{[1]}(t_{2},t_{1}) \circ \mathcal{U}_{[0]}(t_{1},t_{0}) [\rho^{S}\otimes \rho^{E}_{[0]}] \nonumber\\
& = &\mathcal{U}_{[1]}(t_{2},t_{1}) \circ \Lambda_{\rho'^{E}(t_{1})} \circ \mathcal{U}_{[0]}(t_{1},t_{0}) [\rho^{S}\otimes \rho^{E}_{[0]}]
\end{eqnarray}
for all $\rho^{S}\in \mathcal{B}\left(\mathcal{H}^{S}\right)$.  
\end{mydef}
This definition is similar to our definition of no IBTRES, but with an important difference in that the constant map (called replacement channel in Ref.\cite{NMLocal}) in Eq.~(\ref{eqE1}) is applied on the environment state rather than on the system state in our definition. The constant map applied on the environment erases both the information of the initial environment state and the system-environment correlation, but not the system state information. It is a sufficient condition for no IBTRES in our definition, but not a necessary one. For most cases, the dimension of environment state is infinite, and that makes this definition hard to implement experimentally and hard to calculate in practice. In contrast, our definition is operational because the constant map operation is applied only on the system.


\section{General quantum regression formula defined in \cite{NMLocal}}
\label{GQRT}

We compare the general quantum regression formula described in Sec. 3.4.2 and reformulated in Sec. 3.5.2 of \cite{NMLocal} 
to the approach for a Markovian process discussed here.
Note that the combined system-environment state at the initial time is assumed to be factorizable in \cite{NMLocal}.
We first represent the equation of the general quantum regression formula, Eq.(42), in \cite{NMLocal} using our notaition as follows.
The open quantum system dynamics of an $m$-time-step process satisfying the general quantum regression formula if and only if 
there exists a set of unitary maps $\{ \mathcal{W}_{[n]} | n= 0,1,\cdots (m-1) \}$ acting only on the environment state space 
such that
\begin{equation} 
\label{eqE2add}
\rho^{S}_{[m]} = \mathcal{T}^{\rm Ref}_{[m-1]}\mathcal{A}_{[m-1]}\cdots\mathcal{T}^{\rm Ref}_{[1]}\mathcal{A}_{[1]}\mathcal{T}^{\rm Ref}_{[0]}\mathcal{A}_{[0]}[\rho^{S}_{[0]}], 
\end{equation}
where 
\begin{equation}
\label{eqE3add}
\mathcal{T}^{\rm Ref}_{[n]}[\rho^{S}] = {\rm tr_{E}}\mathcal{U}_{[n]}[\rho^{S} \otimes \tilde{\rho}^{E}_{n}],
\end{equation}
with
\begin{equation}
\label{eqE1add}
 \tilde{\rho}^{E}_{n}
=\mathcal{W}_{[n]}[\rho^{E}_{[0]}].
\end{equation}
The general quantum regression formula described in \cite{NMLocal} hinges on the condition to replace the joint state $\rho^{SE}_n$ at time step $n$ by a factorized state $\rho^S_n\otimes\tilde{\rho}^E_n$ in the time-ordered correlation function, conceptually related to the factorization approximation defined in Eq.~(16) of \cite{NMLocal}, where $\tilde{\rho}^E_n$ is defined in Eq.~(\ref{eqE1add}) resulting from unitary map $\mathcal{W}_{[n]}$. 
One may notice that Eq.~(\ref{eqE2add}) is very similar to Eq.~(\ref{MyM}) that describes a Markovian process in our approach.
Thus the general quantum regression formula
is Markovian because Eq.~(\ref{eqE2add}) is in a divisible form.
However, a Markovian process may not satisfy the general quantum regression formula. One can prove this by using the following definition.
\begin{mydef}
\label{postGQRTdef}
An $m$-time-steps process satisfy the extended general quantum regression formula if and only if there exists 
a set of TPCP maps $\{ \mathcal{C}_{[n]} | n= 0,1,\cdots (m-1) \}$
such that 
\begin{equation} 
\label{postDivid}
\rho^{S}_{[m]} = \mathcal{T}^{\rm New}_{[m-1]}\mathcal{A}_{[m-1]}\cdots\mathcal{T}^{\rm New}_{[1]}\mathcal{A}_{[1]}\mathcal{T}^{\rm New}_{[0]}\mathcal{A}_{[0]}[\rho^{S}_{[0]}],
\end{equation}
where
\begin{equation}
  \mathcal{T}^{\rm New}_{[n]}[\rho^{S}] = {\rm tr_{E}}\mathcal{U}_{[n]}[\rho^{S} \otimes \hat{\rho}^{E}_{n}]
  \label{postT}
\end{equation}
with
\begin{equation}
\hat{\rho}^{E}_{n}= \mathcal{C}_{[n]}[\rho^{E}_{[0]}]. 
  \label{postTPCP}
\end{equation}
\end{mydef}
The name of the extended general quantum regression formula in
Definition~\ref{postGQRTdef} comes from the fact that
the set of unitary maps $\{ \mathcal{W}_{[n]} | n= 0,1,\cdots (m-1) \}$ is extended to (replaced by)
a set of TPCP maps $\{ \mathcal{C}_{[n]} | n= 0,1,\cdots (m-1) \}$ 
in the definition of the general quantum regression formula of Eqs.~(\ref{eqE2add})--(\ref{eqE1add}), resulting in Eqs.~(\ref{postDivid})--(\ref{postTPCP}). 
The set of quantum processes described by the original general quantum regression formula is a subset of Definition~\ref{postGQRTdef} because the set of unitary maps $\{ \mathcal{W}_{[n]}\}$ is a subset of TPCP maps $\{ \mathcal{C}_{[n]}\}$.
Because Eq.~(\ref{postDivid}) is in a divisible form, quantum processes described by it are Markovian. But not all the Markovian process described by Eq.~(\ref{postDivid}) satisfy the general quantum regression formula, Eq.~(\ref{eqE2add}).

\begin{Corollary}
  \label{coro5}
  An $m$-time-step process satisfies Definition~\ref{MarkovDef}
  if and only if it satisfies Definition~\ref{postGQRTdef}.
\end{Corollary}
\begin{proof}
  If an $m$-time-step process is Markovian (Definition~\ref{MarkovDef}), it satisfies Eq.~(\ref{MyM}).
To prove that a Markovian process also satisfies Definition~\ref{postGQRTdef},
one needs to find a suitable $\{ \mathcal{C}_{[n]}\}$ such that 
 Eqs.~(\ref{postDivid}) is equivalent to  Eq.~(\ref{MyM}).
One can set 
\begin{equation}
  \mathcal{C}_{[n]} = \Lambda_{\rho^{E}_{[n]}}/{\rm tr}(\rho^{E}_{[n]}),
\end{equation}
a constant map acting on the initial environment state,
where $\rho^{E}_{[n]}= {\rm tr_{S}}(\rho^{SE}_{[n]})$ is the reduced (marginal) environmrnt state at time step $n$.
Therefore, one has
\begin{equation}
\mathcal{C}_{[n]}[\rho^{E}_{[0]}]=\rho^{E}_{[n]}/{\rm tr}(\rho^{E}_{[n]}), 
  \label{TPCPstate}
\end{equation}
and 
\begin{eqnarray}
\mathcal{T}^{\rm New}_{[n]}[\rho^{S}] &=& {\rm tr_{E}}\mathcal{U}_{[n]}[\rho^{S} \otimes \rho^{E}_{[n]}]/{\rm tr}(\rho^{E}_{[n]}) \nonumber \\
                                      &=& \mathcal{L}_{[n]} \left[\rho^{S},\mathcal{A}_{[n-1]} , \mathcal{A}_{[n-2]} ,\cdots, \mathcal{A}_{[0]} \right],
\label{TNew}                                          
\end{eqnarray}
where Eq.~(\ref{Leq}) has been used for the last equality in Eq.~(\ref{TNew}). 
By Proposition~\ref{mainProp}, a Markovian process having neither SECE nor IBTRES, we employ the equation for no IBTRES, i.e., Eq.~(\ref{T1eq}) in Definition~\ref{no-IBTRESdef}, to obtain
\begin{eqnarray}
  \mathcal{T}^{\rm New}_{[n]}[\rho^{S}] &=& \mathcal{T}_{[n]}[\rho^{S}].
\label{TNewTn}
\end{eqnarray}
As a consequence, Eq.~(\ref{postT}) is true because it can be constructed by $\mathcal{T}_{[n]}$.  Equation (\ref{postDivid}) can be also made true by writing out  Eq.~(\ref{postT}) step by step
with the replacement of $\mathcal{T}^{\rm New}_{[n]}[\rho^{S}] = \mathcal{T}_{[n]}[\rho^{S}]$, Eq.~(\ref{TNewTn}), and the resultant equation  is just  Eq.~(\ref{MyM}).
Thus the Markovian $m$-time-steps process satisfies  Eqs.~(\ref{postDivid})--(\ref{postTPCP}) of Definition~\ref{postGQRTdef}  with $\mathcal{T}^{\rm New}_{[n]}[\rho^{S}] = \mathcal{T}_{[n]}[\rho^{S}]$. 

The proof for the reverse statement 
that if an $m$-time-step process satisfies Definition~\ref{postGQRTdef}, then it also satisfies  Definition~\ref{MarkovDef} is straightforward.
As Eq.~(\ref{postDivid}) in Definition~\ref{postGQRTdef} is in a divisible form, the future system state of the process at time step
 $(n+1)$ depends solely on the present state at time step $n$.
So the process is Markovian and satisfies Definition~\ref{MarkovDef}.
\end{proof}
We remark from the above proof that 
to make a full one-to-one correspondence between the
definition of Markovianity in Definition~\ref{MarkovDef}
and the (extended) general quantum regression formula 
is to consider the environment state  $\tilde{\rho}^{E}_{n}$
in Eq.~(\ref{eqE3add}) or  $\hat{\rho}^{E}_{n}$ in Eq.~(\ref{postT})
as the normalized reduced (marginal) state of the
environment at time step $n$, Eq.~(\ref{TPCPstate}),
rather than the environment stae connected to the initial
state by a unitary map, $\mathcal{W}_{[n]}[\rho^{E}_{[0]}]$,
as described in Eq.~(\ref{eqE1add}) or in \cite{NMLocal}.

In summary, we show that by replacing the set of local unitary maps $\{ \mathcal{W}_{[n]}\}$ with a set of TPCP maps $\{ \mathcal{C}_{[n]}\}$,
the general quantum regression formula is a sufficient but
not necessary condition for a Markovian process.
We also show that an $m$-time-steps process satisfying
the resultant extended general quantum regression formula of Definition~\ref{postGQRTdef} satisfies also Definition~\ref{MarkovDef} and Eq.~(\ref{MyM}) that describes a Markovian process, and vice versa.

\section{Proof of Corollary \ref{SECEprop1}}
\label{AppSECEprop1}
We give the proof for Corollary \ref{SECEprop1} here.
Substituting  $\rho^{S}_{[1]}=\mathcal{M}[\mathcal{A}_{[0]}]$ into Eq.~(\ref{Lemma1eq}), we obtain
\begin{equation}
\label{eqF1}
\begin{split}
&T^{i_{1'}j_{1'},i_{0'}j_{0'}}_{\ i_{2}j_{2},\  i_{1}j_{1},\ i_{0}j_{0}}A_{[0];i_{0'}j_{0'}}^{\ \ \ i_{0}j_{0}}A_{[1];i_{1'}j_{1'}}^{\ \ \ i_{1}j_{1}} (\delta_{x_{0'}y_{0'}}M^{\. i_{3} j_{3}}_{x_{0'}y_{0'} ,i_{2'}j_{2'}} A_{[0];i_{3}j_{3}}^{\  \ i_{2'}j_{2'}}) \\ 
&= T^{i_{1'}j_{1'},i_{0'}j_{0'}}_{\ i_{2}j_{2},\  x_{0}y_{0},\ i_{0}j_{0}}A_{[0];i_{0'}j_{0'}}^{\ \ \ i_{0}j_{0}}A_{[1];i_{1'}j_{1'}}^{\ \ \ i_{1}j_{1}} \delta_{x_{0},y_{0}}\cdot
M^{\. i_{3} j_{3}}_{i_{1}j_{1} ,i_{2'}j_{2'}} A_{[0];i_{3}j_{3}}^{\  \ i_{2'}j_{2'}}.
\end{split}
\end{equation}
From Eqs.~(\ref{eq:L}) and (\ref{eq:N}), substituting the expressions $\delta_{x_{0'}y_{0'}}M^{\. i_{3} j_{3}}_{x_{0'}y_{0'} ,i_{2'}j_{2'}} \equiv \left( N_{\ \ \ \ i_{2'}j_{2'}}^{i_{3} j_{3}} \right)$ and 
$T^{i_{1'}j_{1'},i_{0'}j_{0'}}_{\ i_{2}j_{2},\  x_{0}y_{0},\ i_{0}j_{0}} \delta_{x_{0}y_{0}} \equiv \left(L^{i_{1'}j_{1'}, i_{0'}j_{0'}}_{i_{2}j_{2},\ \ \ \ \ \ \ ,i_{0}j_{0}}\right)$ into Eq.~(\ref{eqF1}),
we obtain
\begin{equation}
\label{eqF2}
\begin{split}
&T^{i_{1'}j_{1'},i_{0'}j_{0'}}_{\ i_{2}j_{2},\  i_{1}j_{1},\ i_{0}j_{0}}A_{[0];i_{0'}j_{0'}}^{\ \ \ i_{0}j_{0}}A_{[1];i_{1'}j_{1'}}^{\ \ \ i_{1}j_{1}} \left( N_{\ \ \ \ i_{2'}j_{2'}}^{i_{3} j_{3}} \right) A_{[0];i_{3}j_{3}}^{\  \ i_{2'}j_{2'}} \\ 
&= \left(L^{i_{1'}j_{1'}, i_{0'}j_{0'}}_{i_{2}j_{2},\ \ \ \ \ \ \ ,i_{0}j_{0}}\right)A_{[0];i_{0'}j_{0'}}^{\ \ \ i_{0}j_{0}} A_{[1];i_{1'}j_{1'}}^{\ \ \ i_{1}j_{1}}  
M^{\. i_{3} j_{3}}_{i_{1}j_{1} ,i_{2'}j_{2'}} A_{[0];i_{3}j_{3}}^{\  \ i_{2'}j_{2'}},
\end{split}
\end{equation}
where we have used the notations that the indexes of a tensor are sorted form up to down and right to left in time so that the empty space in the subscripts indicates that the indexes disappear due to a trace over the system Hilbert space $\mathcal{H}^{S_{1}}$.
As  ${A}_{[1]} \in \mathcal{B}\left(\mathcal{H}^{S_{1'}S_{1}}\right)$ is an arbitrary quantum operation, 
we can rewrite Eq.~(\ref{eqF2}) into another process that can be characterized by performing a process tomography via varying ${A}_{[1]}$. As a result, we can  
remove ${A}_{[1]}$ form the both sides of Eq.~(\ref{eqF2}) and obtains
\begin{eqnarray}
\label{eqF3}
&&T^{i_{1'}j_{1'},i_{0'}j_{0'}}_{\ i_{2}j_{2},\  i_{1}j_{1},\ i_{0}j_{0}}A_{[0];i_{0'}j_{0'}}^{\ \ \ i_{0}j_{0}}\left( N_{\ \ \ \ i_{2'}j_{2'}}^{i_{3} j_{3}} \right) A_{[0];i_{3}j_{3}}^{\  \ i_{2'}j_{2'}} \nonumber\\ 
 &=& \left(L^{i_{1'}j_{1'}, i_{0'}j_{0'}}_{i_{2}j_{2},\ \ \ \ \ \ \ ,i_{0}j_{0}}\right)A_{[0];i_{0'}j_{0'}}^{\ \ \ i_{0}j_{0}}
M^{\. i_{3} j_{3}}_{i_{1}j_{1} ,i_{2'}j_{2'}} A_{[0];i_{3}j_{3}}^{\  \ i_{2'}j_{2'}}.
\end{eqnarray}	
In contrast, the two ${A}_{[0]}$'s appearing in either Eq.~(\ref{eqF2}) or Eq.~(\ref{eqF3}) cannot be removed from the both sides of the equation.
This is because the two ${A}_{[0]}\otimes {A}_{[0]} \in \mathcal{B}\left(\mathcal{H}^{S_{0'}S_{0}}\right)\otimes\mathcal{B}\left(\mathcal{H}^{S_{2'}S_{3}}\right)$ 
are the same operation and thus can not be varied independently to form a complete basis required for constructing a process without ${A}_{[0]}$'s in  Eq.~(\ref{eqF3}) 
by process tomography. 
Thus, let 
\begin{equation}
\label{eqF4}
\mathcal{A}_{[0]}=\alpha\mathcal{A}+\beta\mathcal{A}',
\end{equation}
where $\alpha$ and $\beta$ are positive real numbers with $0\leq\alpha+\beta\leq 1$, $\mathcal{A}$ and $\mathcal{A}'$ are arbitrarily chosen quantum operations satisfying Eq.~(\ref{eqF4}). Substituting Eq.~(\ref{eqF4}) in tensor form into Eq.~(\ref{eqF3}) leads to
\begin{eqnarray}
\label{eqF5}
&& T^{i_{1'}j_{1'},i_{0'}j_{0'}}_{\ i_{2}j_{2},\  i_{1}j_{1},\ i_{0}j_{0}}\left( \alpha A_{i_{0'}j_{0'}}^{i_{0} j_{0}}+\beta A'^{i_{0} j_{0}}_{i_{0'}j_{0'}} \right) \nonumber\\
 &&\cdot\left( N_{\ \ \ \ i_{2'}j_{2'}}^{i_{3} j_{3}} \right) \left(\alpha A^{i_{2'}j_{2'}}_{i_{3}j_{3}}+\beta A'^{i_{2'}j_{2'}}_{i_{3}j_{3}} \right)\nonumber\\ 
 &=&\left(L^{i_{1'}j_{1'}, i_{0'}j_{0'}}_{i_{2}j_{2},\ \ \ \ \ \ \ ,i_{0}j_{0}}\right)\left( \alpha A_{i_{0'}j_{0'}}^{i_{0} j_{0}}+\beta A'^{i_{0} j_{0}} _{i_{0'}j_{0'}}\right) \nonumber\\
&&\cdot
M^{\. i_{3} j_{3}}_{i_{1}j_{1} ,i_{2'}j_{2'}} \left(\alpha A^{i_{2'}j_{2'}}_{i_{3}j_{3}}+\beta A'^{i_{2'}j_{2'}}_{i_{3}j_{3}} \right).
\end{eqnarray}	
The terms with the factor $\alpha \beta$ on the both sides of Eq.~(\ref{eqF5}) must equal to each other as $\mathcal{A}_{[0]}$ and the decomposition of Eq.~(\ref{eqF4}) are arbitrary. So we have
\begin{eqnarray}
\label{alba}
&&\alpha \beta \cdot T^{i_{1'}j_{1'},i_{0'}j_{0'}}_{\ i_{2}j_{2},\  i_{1}j_{1},\ i_{0}j_{0}}\left( N_{\ \ \ \ i_{2'}j_{2'}}^{i_{3} j_{3}} \right) A_{i_{0'}j_{0'}}^{i_{0} j_{0}} A'^{i_{2'}j_{2'}}_{i_{3}j_{3}}\nonumber\\
&&+ \alpha \beta \cdot T^{i_{1'}j_{1'},i_{3}j_{3}}_{\ i_{2}j_{2},\  i_{1}j_{1},\ i_{2'}j_{2'}}\left( N_{\ \ \ \ \  i_{0}j_{0}}^{i_{0'} j_{0'}} \right)  A'^{i_{0} j_{0}}_{i_{0'}j_{0'}} A^{i_{2'}j_{2'}}_{i_{3}j_{3}}\nonumber\\
&=& \alpha \beta\cdot \left(L^{i_{1'}j_{1'}, i_{0'}j_{0'}}_{i_{2}j_{2},\ \ \ \ \ \ \ ,i_{0}j_{0}}\right)M^{\. i_{3} j_{3}}_{i_{1}j_{1} ,i_{2'}j_{2'}}  A_{i_{0'}j_{0'}}^{i_{0} j_{0}}A'^{i_{2'}j_{2'}}_{i_{3}j_{3}} \nonumber\\
&&\hspace{-0.2cm}+\alpha \beta\cdot\left(L^{i_{1'}j_{1'}, i_{3}j_{3}}_{i_{2}j_{2},\ \ \ \ \ \ \ ,i_{2'}j_{2'}}\right)M^{\. i_{0'} j_{0'}}_{i_{1}j_{1} ,i_{0}j_{0}} A'^{i_{0} j_{0}} _{i_{0'}j_{0'}} A^{i_{2'}j_{2'}}_{i_{3}j_{3}}.
\end{eqnarray}	
We can relabel the dummy variables in $A$ and $A'$,e.g., $A^{i_{2'}j_{2'}}_{i_{3}j_{3}} \rightarrow A_{i_{0'}j_{0'}}^{i_{0} j_{0}}$ and  $A'^{i_{0} j_{0}}_{i_{0'}j_{0'}} \rightarrow  A'^{i_{2'}j_{2'}}_{i_{3}j_{3}}$. The respective dummy variable in $T$,$L$,$M$ and $N$ are also changed. As a result, Eq.~(\ref{alba}) becomes
\begin{eqnarray}
\label{unremove}
&& T^{i_{1'}j_{1'},i_{0'}j_{0'}}_{\ i_{2}j_{2},\  i_{1}j_{1},\ i_{0}j_{0}}\left( N_{\ \ \ \ i_{2'}j_{2'}}^{i_{3} j_{3}} \right) A_{i_{0'}j_{0'}}^{i_{0} j_{0}} A'^{i_{2'}j_{2'}}_{i_{3}j_{3}} \nonumber\\
&&+  T^{i_{1'}j_{1'},i_{3}j_{3}}_{\ i_{2}j_{2},\  i_{1}j_{1},\ i_{2'}j_{2'}}\left( N_{\ \ \ \ \ i_{0}j_{0}}^{i_{0'} j_{0'}} \right)  A^{i_{0} j_{0}}_{i_{0'}j_{0'}} A'^{i_{2'}j_{2'}}_{i_{3}j_{3}}\nonumber\\
&=&  \left(L^{i_{1'}j_{1'}, i_{0'}j_{0'}}_{i_{2}j_{2},\ \ \ \ \ \ \ ,i_{0}j_{0}}\right)M^{\. i_{3} j_{3}}_{i_{1}j_{1} ,i_{2'}j_{2'}}  A_{i_{0'}j_{0'}}^{i_{0} j_{0}}A'^{i_{2'}j_{2'}}_{i_{3}j_{3}} \nonumber\\
&&+\left(L^{i_{1'}j_{1'}, i_{3}j_{3}}_{i_{2}j_{2},\ \ \ \ \ \ \ ,i_{2'}j_{2'}}\right)M^{\. i_{0'} j_{0'}}_{i_{1}j_{1} ,i_{0}j_{0}} A^{i_{0} j_{0}} _{i_{0'}j_{0'}} A'^{i_{2'}j_{2'}}_{i_{3}j_{3}}.
\end{eqnarray}	
We can then 
remove $ A^{i_{0} j_{0}} _{i_{0'}j_{0'}} A'^{i_{2'}j_{2'}}_{i_{3}j_{3}}$ from the both sides of Eq.~(\ref{unremove}) and obtain
\begin{eqnarray}
\label{noCEt}
&& T^{i_{1'}j_{1'},i_{0'}j_{0'}}_{\ i_{2}j_{2},\  i_{1}j_{1},\ i_{0}j_{0}}\left( N_{\ \ \ \ i_{2'}j_{2'}}^{i_{3} j_{3}} \right) \nonumber \\
&&+  T^{i_{1'}j_{1'},i_{3}j_{3}}_{\ i_{2}j_{2},\  i_{1}j_{1},\ i_{2'}j_{2'}}\left( N_{\ \ \ \ \ i_{0}j_{0}}^{i_{0'} j_{0'}} \right)  \nonumber \\
&=&  \left(L^{i_{1'}j_{1'}, i_{0'}j_{0'}}_{i_{2}j_{2},\ \ \ \ \ \ \ ,i_{0}j_{0}}\right)M^{\. i_{3} j_{3}}_{i_{1}j_{1} ,i_{2'}j_{2'}} \nonumber  \\
&&+ \left(L^{i_{1'}j_{1'}, i_{3}j_{3}}_{i_{2}j_{2},\ \ \ \ \ \ \ ,i_{2'}j_{2'}}\right)M^{\. i_{0'} j_{0'}}_{i_{1}j_{1} ,i_{0}j_{0}}. 
\end{eqnarray}
Adding the basis vectors $\ket{i_{3}}_{3}\bra{j_{3}}\otimes\ket{i_{2'}}_{2'}\bra{j_{2'}}\otimes\ket{i_{2}}_{2}\bra{j_{2}}\otimes\ket{i_{1'}}_{1'}\bra{j_{1'}}\otimes\ket{i_{1}}_{1}\bra{j_{1}}\otimes\ket{i_{0'}}_{0'}\bra{j_{0'}}\otimes\ket{i_{0}}_{0}\bra{j_{0}}$ to Eq.~(\ref{noCEt}), we can write down the matrix form of Eq.~(\ref{noCEt}) as 
\begin{equation}
\label{eq:noCE}
 \begin{split} 
& N \otimes T + \mathcal{S}_{30'}\circ\mathcal{S}_{2'0} \left(N \otimes T \right) \\
&= \mathcal{S}_{32}\circ\mathcal{S}_{2'1'} (L \otimes M)+ \mathcal{S}_{32}\circ\mathcal{S}_{2'1'}\circ \mathcal{S}_{20'}\circ\mathcal{S}_{1'0} (L \otimes M).
\end{split}
\end{equation}
By writing Eq.~(\ref{eq:noCE}) in a more concise expression, we finally arrive at
\begin{eqnarray}
&&\left( \mathcal{I} + \mathcal{S}_{30'} \circ \mathcal{S}_{2'0} \right)(N \otimes T) \nonumber\\
&=& \mathcal{S}_{32} \circ \mathcal{S}_{2'1'} \circ \left( \mathcal{I} + \mathcal{S}_{20'} \circ \mathcal{S}_{1'0} \right)(L \otimes M)
\end{eqnarray}
which is just Eq.~(\ref{CE_matrix0}) in the main text.

The proof for the reverse statement can be easily carried out as the derivation steps and equations in the proof are reversible.

\section{Proof of Corollary \ref{SECEco}}
\label{AppSECEcoro1}

The proof of Eq.~(\ref{P1Markov}) for no SECE in the first time step is the same as the proof for Proposition \ref{prop1}. So we focus on the proof for  Eq.~(\ref{CE_matrix}) in the second time step.
The result of Eq.~(\ref{P1Markov}) indicates that
the initial system-environment state is effectively uncorrelated, and thus we can use the reduced process tensor $\tilde{T}$ to express equation for a two-time-step process that has no SECE.
Using the trace  preserving condition
$\tilde{M}^{i_{0'}j_{0'}}_{i_{1}j_{1}} \delta_{i_{1}j_{1}} = \delta_{i_{0'}j_{0'}}$ for the definition
$\left(N^{i_{0'}j_{0'}}_{\ \ \ \ \  i_{0}j_{0}}\right) = \rho^{S}_{[0];i_{0}j_{0}} \tilde{M}^{i_{0'}j_{0'}}_{i_{1}j_{1}} \delta_{i_{1}j_{1}} = \rho^{S}_{[0];i_{0}j_{0}} \delta_{i_{0'}j_{0'}}$,
one obtains
\begin{equation}
  \label{eq:Ntr}
  \left(N^{i_{0'}j_{0'}}_{\ \ \ \ \  i_{0}j_{0}}\right) A'^{i_{0}j_{0}}_{i_{0'}j_{0'}} =\rho^{S}_{[0];i_{0}j_{0}} \delta_{i_{0'}j_{0'}}A'^{i_{0}j_{0}}_{i_{0'}j_{0'}} =\rho'_{i_{0'}j_{0'}} \delta_{i_{0'}j_{0'}}.
 \end{equation}
Therefore, employing Eq.~(\ref{eq:Ntr}) and substituting the definitions of 
\begin{eqnarray}
\rho^{S}_{[0];i_{0}j_{0}} A^{i_{0}j_{0}}_{i_{0'}j_{0'}} &\equiv& \rho_{i_{0'}j_{0'}}, \\
\rho^{S}_{[0];i_{0}j_{0}} A'^{i_{0}j_{0}}_{i_{0'}j_{0'}}& \equiv& \rho'_{i_{0'}j_{0'}}
\end{eqnarray}
into Eq.~(\ref{unremove}), we obtain
\begin{equation}
\label{eqG1}
\begin{split}
&  \tilde{T}^{i_{1'}j_{1'},i_{0'}j_{0'}}_{\ i_{2}j_{2},\  i_{1}j_{1}	} \rho_{i_{0'}j_{0'}}\rho'_{i_{3}j_{3}} \delta_{i_{3}j_{3}}
+  \tilde{T}^{i_{1'}j_{1'},i_{3}j_{3}}_{\ i_{2}j_{2},\  i_{1}j_{1}} \rho'_{i_{3}j_{3}}\rho_{i_{0'}j_{0'}} \delta_{i_{0'}j_{0'}} \\
&=  \tilde{L}^{i_{1'}j_{1'},i_{0'}j_{0'}}_{i_{2}j_{2}} \tilde{M}^{\. i_{3} j_{3}}_{i_{1}j_{1}}  \rho_{i_{0'}j_{0'}}\rho'_{i_{3}j_{3}} +\tilde{L}^{i_{1'}j_{1'},  i_{3}j_{3}}_{i_{2}j_{2} }\tilde{M}^{\. i_{0'} j_{0'}}_{i_{1}j_{1}} \rho_{i_{0'}j_{0'}}  \rho'_{i_{3}j_{3}}.
\end{split}
\end{equation}	
Removing $\rho_{i_{0'}j_{0'}}\rho'_{i_{3}j_{3}}$ from the both sides of Eq.~(\ref{eqG1}) and relabeling $(i_{3},j_{3}) \rightarrow (i_{0},j_{0})$, we arrive at 
\begin{eqnarray}
\label{CE_matrixT}
 &&\tilde{T}^{i_{1'}j_{1'},i_{0'}j_{0'}}_{\ i_{2}j_{2},\  i_{1}j_{1}	} \delta_{i_{0}j_{0}}
+  \tilde{T}^{i_{1'}j_{1'},i_{0}j_{0}}_{\ i_{2}j_{2},\  i_{1}j_{1}} \delta_{i_{0'}j_{0'}}\\
&=&  \tilde{L}^{i_{1'}j_{1'},i_{0'}j_{0'}}_{i_{2}j_{2}} \tilde{M}^{\. i_{0} j_{0}}_{i_{1}j_{1}}  +\tilde{L}^{i_{1'}j_{1'},  i_{0}j_{0}}_{i_{2}j_{2} }\tilde{M}^{\. i_{0'} j_{0'}}_{i_{1}j_{1}}.
\end{eqnarray}	
Adding the basis vectors $\ket{i_{2'}}_{2'}\bra{j_{2'}}\otimes\ket{i_{2}}_{2}\bra{j_{2}}\otimes\ket{i_{1'}}_{1'}\bra{j_{1'}}\otimes\ket{i_{1}}_{1}\bra{j_{1}}\otimes\ket{i_{0'}}_{0'}\bra{j_{0'}}\otimes\ket{i_{0}}_{0}\bra{j_{0}}$ to Eq.~(\ref{CE_matrixT}), we finally obtain the matrix form of Eq.~(\ref{CE_matrixT}):
\begin{equation}
(\mathcal{I} + \mathcal{S}_{0'0})[\tilde{T} \otimes I ] = (\mathcal{I} + \mathcal{S}_{0'0}) \circ \mathcal{S}_{10'} [\tilde{L} \otimes \tilde{M}],
\end{equation}
which is just Eq.~(\ref{CE_matrix}) in the main text.

The proof for the reverse statement is easy as the derivation procedure is reversible.

\end{appendix}

\end{document}